%% file: LuxburgRadlHein.tex
\documentclass[10pt]{article}
\pdfoutput=1
\usepackage[round,comma]{natbib}
\usepackage{geometry}
\usepackage{subfig}

\input{ules_latex_defaults.tex}

\let\star\undefined
\newcommand{\star}{$(\bigstar)$}

\newcommand{\ip}[2]{\ensuremath{\left\langle #1, #2 \right\rangle}}

\newcommand{\inner}[1]{\left\langle#1\right\rangle}

\newcommand{\nmin}{N_{\min}}

\DeclareMathOperator{\neigh}{Neigh}
\def\Exp{\mathbb{E}}

\newcommand{\pbetween}{p_{\text{between}}}
\newcommand{\pwithin}{p_{\text{within}}}

\begin{document}

\title{Hitting and commute times in large graphs are often misleading}
\author{Ulrike von Luxburg\\
Max Planck Institute for Biological Cybernetics, T{\"u}bingen, Germany
\and
Agnes Radl\\
Max Planck Institute for Biological Cybernetics, T{\"u}bingen, Germany
\and
 Matthias Hein\\
Saarland University, Saarbr{\"u}cken, Germany
}

\maketitle

\begin{abstract} 
Next to the shortest path distance, the second most popular distance
function between vertices in a graph is the commute distance
(resistance distance).   For two vertices $u$ and $v$, the hitting time $H_{uv}$
  is the expected time it takes a random walk to travel from $u$ to
  $v$. The commute time is its symmetrized version $C_{uv} = H_{uv} +
  H_{vu}$. In our paper we study the behavior of hitting times and
  commute distances when the number $n$ of vertices in the graph is
  very large. We   prove that as $n \to \infty$,  under mild assumptions, hitting times and commute
  distances converge to expressions that do not take into account the
  global structure of the graph at all. Namely, the hitting time
  $H_{uv}$ converges to $1/d_v$ and the commute time to $1/d_u +
  1/d_v$ where $d_u$ and $d_v$ denote the degrees of vertices $u$ and
  $v$. In these cases, the hitting and commute times are misleading in
  the sense that they do not provide information about the 
  structure of the graph. We focus on two major classes of random graphs: random
  geometric graphs (\knn-graphs, $\eps$-graphs, Gaussian similarity
  graphs) and random graphs with given expected degrees (in
  particular, \er graphs with and without planted partitions). %
\end{abstract}

\section{Introduction} \label{sec-intro}

Given an undirected, weighted graph $G = (V,E)$ with $n$ vertices, the
commute distance between two vertices $u$ and $v$ is defined as the
expected time it takes the natural random walk starting in vertex $u$ to travel
to vertex $v$ and back to $u$. It is equivalent (up to a constant) to
the resistance distance, which interprets the graph as an electrical
network and defines the distance between vertices $u$ and $v$ as the
effective resistance between these vertices. See below for exact
definitions and 
\citet{DoySne84,KleRan93,XiaGut03, FouPirRenSae06} 
for background reading. 
The commute distance is very popular in many
different fields of computer science and beyond. As examples consider 
the tasks of graph embedding
\citep{Guattery98,SaeFouYenDup04,QiuHan06_spr,WittmannEtal09}, 
graph sparsification \citep{SpiSri08}, 
social network analysis \citep{LibKle03},
proximity search \citep{SarMooPra08}, 
collaborative filtering \citep{FouPirRenSae06},
clustering \citep{YenEtal05},
semi-supervised learning \citep{ZhoSch04},
dimensionality reduction \citep{HamEtal04}, 
image processing \citep{QiuHan05}, 
graph labeling \citep{HerPon06,CesGenVit09}, 
theoretical computer science 
(\citealp{%
AleliunasEtal79,
ChandraEtal89,
AviErc07,
CooFri03,
CooFri05,
CooFri07,
CooFri09}), 
and various applications in chemometrics and bioinformatics 
\citep{KleRan93,Ivanciuc00,Fowler02,Roy04,GuillotEtal09}. \\

The commute distance has many nice properties, both from a theoretical
and a practical point of view.   It is a Euclidean distance function and can be computed in
closed form. As opposed to the shortest path distance, it takes into
account all paths between $u$ and $v$, not just the shortest one. As a
rule of thumb, the more paths connect $u$ with $v$, the smaller their
commute distance becomes. Hence it supposedly satisfies
the following, highly desirable property:

\begin{quotation} \noindent{\bf Property \star: }
Vertices in the same ``cluster'' of the graph have a
small commute distance, whereas vertices in different clusters
of the graph have a large commute distance to each other. 
\end{quotation}

Consequently, the commute distance is considered a convenient tool
to encode the cluster structure of the graph. \\

In this paper we study how the commute distance behaves when the size
of the graph increases. Our main result is that if the graph is large
enough, then in many graphs the hitting times and commute
distances can be approximated by an extremely simple formula with very
high accuracy. Namely, denoting by $H_{uv}$ the expected hitting time
and by $C_{uv}$ the commute distance between two vertices $u$ and $v$,
by $d_u$ the degree of vertex $u$, and by $\vol(G)$ the volume of the
graph, we show that if the graph gets large enough, for all vertices
$u \neq v$,
\ba %
\frac{1}{\vol(G)}H_{uv}\approx \frac{1}{d_v}  
&& \text{ and } &&
\frac{1}{\vol(G)}C_{uv}\approx \frac{1}{d_u} + \frac{1}{d_v}. 
\ea
On the one hand, we prove these results for arbitrary fixed, large graphs
(Proposition \ref{prop-approx}). Here the quality of the approximation
depends on geometric quantities describing the graph (such as minimal
and maximal degrees, the spectral gap, and so on). The main part of
the paper prove that results hold with probability tending to 1, as $n
\to \infty$, in all major classes of random graphs: random geometric
graphs ($k$-nearest neighbor graphs, $\eps$-graphs,
and Gaussian similarity graphs)  %
and for random graphs with given expected degrees (in particular, also
\er
graphs with and without planted partitions). As a rule of thumb, our
approximation results hold whenever
 the minimal degree in the graph
increases with $n$ (for example, as $\log (n)$ in random geometric
graphs or as $\log^2(n)$ in random graphs with given expected degrees). \\

In order to make our results as accessible as possible to a wide range
of computer scientists, we present two
different strategies to prove our results: one based on flow arguments
on electrical networks and another based on spectral
arguments. While the former approach leads to tighter bounds, the
latter is more general. An important step on the way is that we prove
bounds on the spectral gap in all classes of random geometric
graphs. This is interesting by itself as the spectral gap governs many
important properties and processes on graphs. In this generality, the
bounds on the spectral
gaps are new. \\

Our results have important consequences. \\

{\bf Hitting and commute times in large graphs are often misleading.}  On the negative side, our approximation result shows
that contrary to popular belief, the commute distance does not take
into account any global properties of the data, at least if the graph
is ``large enough''. It just considers the local density (the degree
of the vertex) at the two vertices, nothing else. The resulting large
sample commute distance $dist(u,v) = 1/d_u + 1/d_v$ is completely
meaningless as a distance on a graph.  For example, all data points
have the same nearest neighbor (namely, the vertex with the largest
degree), the same second-nearest neighbor (the vertex with the
second-largest degree), and so on. In particular, one of the main motivations
to use the commute distance, Property \star, no longer holds when the
graph becomes large enough.  Even more disappointingly, computer
simulations show that $n$ does not even need to be very large before
\star\ breaks down. Often, $n$ in the order of 1000 is already enough
to make the commute distance very close to its approximation
expression.  This effect is even stronger if the dimensionality of the
underlying data space is large.  Consequently, even on moderate-sized
graphs, the use of the raw commute distance should be discouraged. \\

{\bf Efficient computation of approximate commute distances.} 
In some applications the commute distance is not used as
a distance function, but as a tool to encode the connectivity
properties of a graph, for example in graph sparsification \citep{SpiSri08}
or when computing bounds on mixing or cover times \citep{%
  AleliunasEtal79, ChandraEtal89, AviErc07,CooFri09} or graph labeling
\citep{HerPon06,CesGenVit09}. To obtain the commute distance between all points in a graph one has to
compute the pseudo-inverse of the graph Laplacian matrix, an operation
of time complexity $O(n^3)$. This is prohibitive in large graphs. To
circumvent the matrix inversion, several approximations of the commute
distance have been suggested in the literature
\citep{SpiSri08,SarMoo07,Brand05}. Our results lead to a
much simpler and well-justified way of approximating the commute
distance on large random geometric graphs. \\

We start our paper with Section \ref{sec-intuition} 
that tries to convey our main results and techniques on a very high
level. Then, after introducing general definitions
and notation (Section \ref{sec-notation}), we present our main results in
Section \ref{sec-main}. This section is divided into two parts (flow
based part and spectral part). All  proofs are
presented in Sections \ref{sec-proofs-flow} and
\ref{sec-proofs-spectral}. A final discussion can be found in Section
\ref{sec-discussion}. For the convenience of the reader, some basic facts on random geometric graphs are
presented in the appendix. Parts of this work is built on our conference
paper \citet{LuxRadHei10}. \\

\section{Intuition about our results and proofs} \label{sec-intuition}
\begin{figure}[t]
\bc
\includegraphics[width=0.7\textwidth]{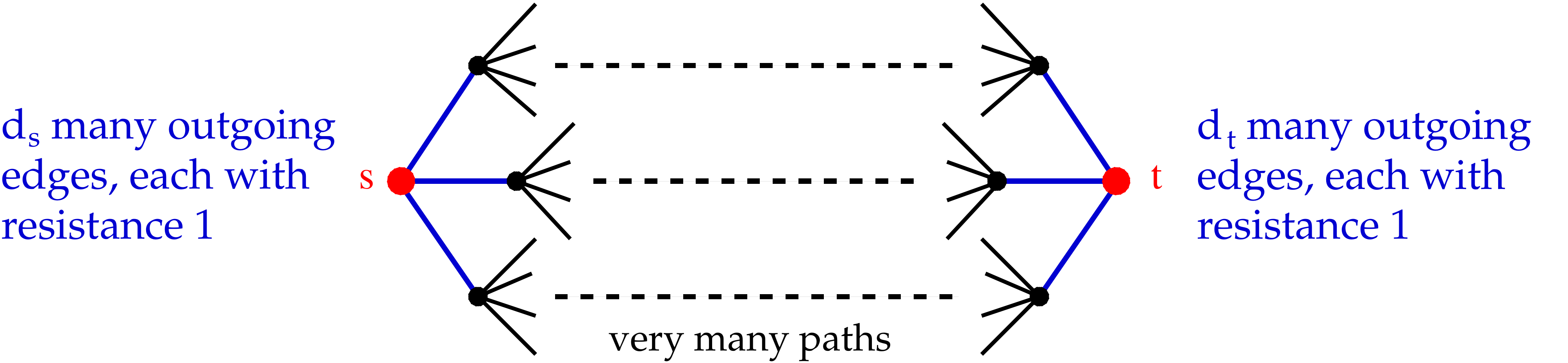}
\ec
\caption{\em Electrical network intuition: The effective resistance
  between $s$ and $t$ is dominated by the edges adjacent to $s$ and $t$.}
\label{fig-electrical}
\end{figure}

Before diving into technicalities, we would like to present our
results in an intuitive, non-technical way. Readers interested in
crisp theorems are encouraged to skip this section right away. \\

Informally the main result of our paper is the following: \\

{\bf Main result: } {\em 
Consider a ``large'' graph that is ``reasonably strongly'  connected. In
such a graph, the hitting times and commute distances between any two
vertices $u$ and $v$ can be
approximated by the simple expressions 
\ba %
\frac{1}{\vol(G)}H_{uv}\approx \frac{1}{d_v}  
&& \text{ and } &&
\frac{1}{\vol(G)}C_{uv}\approx \frac{1}{d_u} + \frac{1}{d_v}. 
\ea
}

In this section we want to present some intuitive arguments to
understand why  this makes sense. In order to 
show a broad picture and to make our results accessible to a general audience, we are going to present two completely different
approaches in our paper.

\subsection{Electrical network intuition}

Consider an unweighted graph as an electrical network where each edge has
resistance 1. We want to compute the effective resistance between two
fixed vertices $s$ and $t$ by exploiting the electrical laws. Resistances
in series add up, that is for two resistances $R_1$ and $R_2$ in
series we get the overall resistance $R=R_1 + R_2$. Resistances in
parallel lines satisfy $1/R = 1/R_1 + 1/R_2$. Now consult the
unweighted electrical network in Figure \ref{fig-electrical}. Consider the vertex $s$ and
all edges from $s$ to its $d_s$ neighbors. The resistance ``spanned''
by these $d_s$ parallel edges satisfies $1/R = \sum_{i=1}^{d_s} 1 = d_s$,
that is $R = 1/d_s$. Similarly for $t$. Between the neighbors of $s$
and the ones of $t$ there are very many paths. It turns out that the
contribution of these paths to the resistance is negligible
(essentially, we have so many wires between the two neighborhoods that
electricity can flow nearly freely). So the overall effective
resistance between $s$ and $t$ is dominated by the edges adjacent to
$i$ and $j$ with contributions $1/d_s + 1/d_t$. \\

The main theorems derived from the electrical network approach are
Theorems \ref{th-main-eps} and \ref{th-main-knn}.  In order to prove
them, we bound the electrical resistance between two vertices using
flow arguments. The overall idea is that we construct a unit flow
between $s$ and $t$ that uses as many paths as possible. From the
technical side, this approach has the advantage that we can throw away
irrelevant parts of the graph --- we can concentrate on a ``valid
region'' that contains $s$, $t$, and many paths between $s$ and
$t$. For this reason, we need less assumptions on the geometry of the
underlying space ``close to its boundary''. We explicitly construct
such flows for random geometric graphs. The idea is to place a grid on
the underlying space and control the flow between different
cells of the grid. \\

As far as we can see, this technique can only be used to bound the resistance distance
$R_{ij}$, it does not work for the individual hitting times $H_{ij}$ or $H_{ji}$. \\

\subsection{Random walk intuition}

\begin{figure}[t]
\bc
\includegraphics[width=0.3\textwidth]{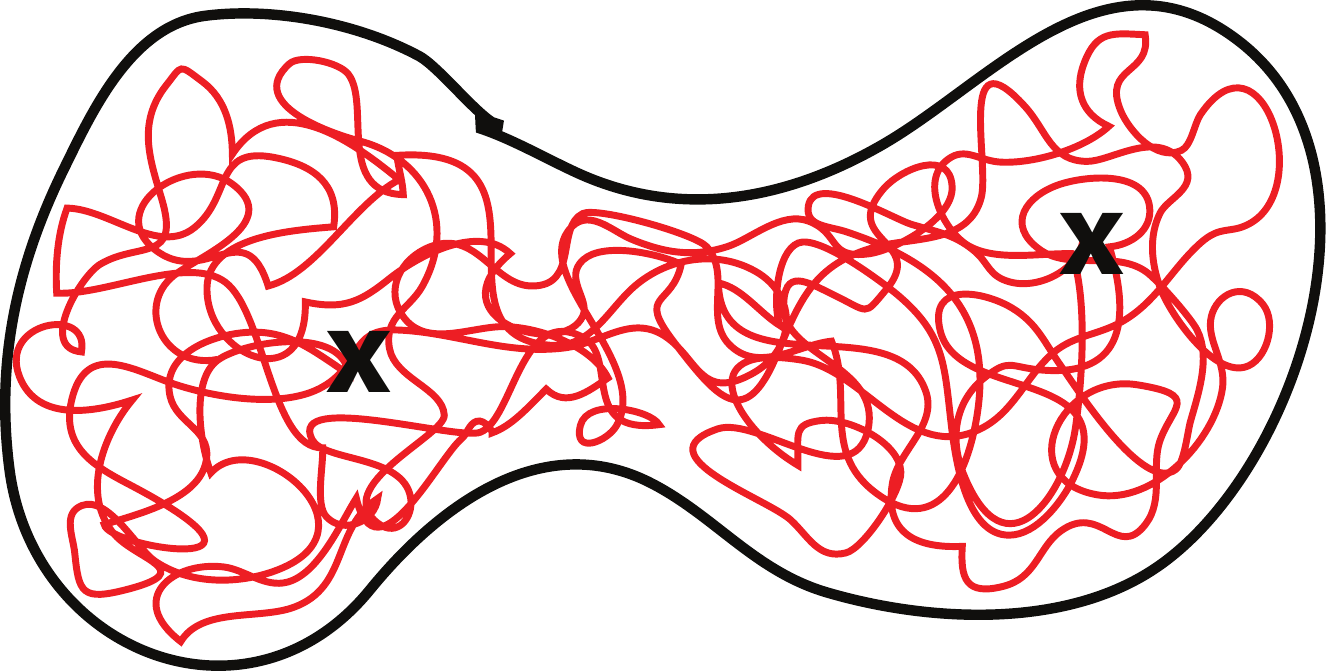}
\ec
\caption{\em Random walk intuition: Between its start and target vertex
  (black crosses), the random walk wanders around so long that by the time
it finally arrives at its target it has already ``forgotten'' where it started from. }
\label{fig-rw}
\end{figure}

Another approach to understand our convergence results is based on
random walks. Essentially, our results for the hitting times $H_{uv}$ 
say that regardless at which vertex $u$ we start, the time to hit
vertex $v$ just depends on the degree of $v$. What happens is that as
the graph gets large, the random walk can explore so many paths that by the
time it is close to $v$ it ``has forgotten'' where it came from
(cf. Figure \ref{fig-rw}). This
is why the hitting time does not depend on $u$. Once the
random walk is in the vicinity of $v$, the question is just whether it
exactly hits $v$ or whether it passes close to $v$ without hitting
it. Intuitively, the likelihood to hit $v$ is inversely proportional to the density of the
graph close to $v$: if there are many edges in the neighborhood of $v$, then it is
easier to hit $v$ than if there are only few edges. This is how
the inverse degree comes into play. \\

Stated slightly differently, the random walk has already mixed before
it hits $v$. For this reason, the hitting time does not depend on
$u$. All that is left is some component depending on $v$. Notably, this component exactly
coincides with the mean return time of $v$ (the expected time it takes
a random walk that starts at $v$ to return to $v$), which is given as
$\vol(G) / d_v$. \\

In the light of our explanation it is reasonable to expect that the quality of our
approximation depends on the mixing time of the random walk, and the
latter is known to be governed by the size of the spectral gap, in particular the
quantity $1 - \lambda_2$ (see below for exact definitions). Indeed, we will
see in our Key Proposition
\ref{prop-approx} that $1-\lambda_2$ is exactly the quantity that governs the
deviation bound for the hitting and commute times. If $1-\lambda_2$ is small, then the graph is
too well-clustered, has a large mixing time, and our approximation
guarantee gets worse. \\

The spectral approach leads to the main theorems in Section
\ref{sec-main-spectral}. We first have to express 
the commute time in terms of a spectral representation of the
graph (Proposition \ref{prop-approx}). 
To make use of this proposition 
we need a lower bound on the
spectral gap $1-\lambda_2$ of the graph. \\
To bound the spectral
gap in random geometric graphs we use path-based arguments as well, namely we use the canonical
path technique of \citet{DiaStr91}. Here one has to construct a set
of ``canonical paths'' between each pair of vertices in the graph. The
goal is to distribute these paths ``as well as possible'' over the
graph.  As in the case above we use a grid to control the paths
between different cells of this grid. This is very reminiscent of the
technique described above. However, an important difference is that we
now need to consider paths between all pairs of points (we have to
bound the spectral gap of the whole graph) instead of just paths
between $s$ and $t$. In the language of flows, instead of looking at a
unit flow from $s$ to $t$ we would have to use
multi-commodity flows between all pairs of vertices instead of a
single flow from $s$ to $t$ (cf.\ \citealp{Sinclair92,DiaSal93}). 
For this reason, we need stronger assumptions on
the geometry of the underlying space.\\

For the case of random graphs with expected degrees, we build on
results about the spectral gap from the literature. As in the
electrical network approach, we need to ensure that these graphs are
``strongly enough'' connected. This will be achieved by requiring that
the minimal vertex degree in the graph is ``large enough'' (with
respect to the number $n$ of vertices). Our results hold for any
arbitrary degree distribution, as soon as the minimal degree grows
slowly with $n$.

The advantage of the spectral approach is that it is very general. It
works for any kind of graph, and as opposed to the electrical network
approach can also be used to treat the hitting times directly. The
technical disadvantage is that we cannot ``throw away'' irrelevant
parts of the graph as in the electrical network approach (because no
part of the graph is irrelevant to
the gap), leading to slightly worse bounds. \\

\subsection{General limitations}

There are two major limitations to our results: 

\bi 
\item 
Our approximation results only hold if the graph is ``reasonably
strongly'' connected and does not have too
large a bottleneck. This ensures that the overall behavior of the
commute distance cannot be dominated by a single edge. We can see this in both approaches. In the
electrical network approach, the argument that 
``electricity can nearly flow without resistance'' on the ``many paths'' breaks
down if there is a strong bottleneck between $u$ and $v$ which all electricity has to
pass. In the spectral approach, a strong bottleneck leads to a small
spectral gap, and then the bounds become meaningless as well. 

\item Our results only hold if the minimal degree in the graph is
``reasonably large'', compared to the number $n$ of vertices. For
example, in the
random graph models the minimal degree has to grow slowly with $n$,
say as $\log n$. This is to ensure that there are no single
vertices that can have extremely high influence on the commute
distance. \\

The downside  
of this
condition is that our results do not hold for power law graphs in
which the smallest degree is constant. 

\ei

As presented in this intuitive section, it nearly sounds as if our
results were obvious. Indeed, in hindsight they seem to be obvious,
and this is part of why we like our results so much: they were very
surprising when we found them, but can be made plausible to a wide
range of people. We would like to stress that all these results were
not known before our work, and that the ``intuitive explanations''
have to be seen as the succession of our technical work. In
particular, the technical work presented in the rest of this paper
makes explicit all the sloppy terms like ``reasonably connected'' and
``large enough''.

\section{General setup, definitions and notation}\label{sec-notation}

We consider undirected graphs $G = (V,E)$ that are connected and not bipartite. By $n$ we denote the
number of vertices. The adjacency matrix is denoted by $W :=
(w_{ij})_{i,j=1, \hd, n}$. In case the graph is weighted, this matrix
is also called the weight matrix. All weights are assumed to be
non-negative. The minimal and maximal weights in the graph are
denoted by $w_{\min}$ and $w_{\max}$. By $d_i := 
\sum_{j=1}^n w_{ij}$ we denote the degree of vertex $v_i$. The
diagonal matrix $D$ with diagonal entries
$d_1, \hd, d_n$ is called the {\em degree matrix}, the minimal
and maximal degrees are denoted $d_{\min}$ and $d_{\max}$.  The {\em unnormalized graph
  Laplacian} is given as $L:=D-W$, the normalized one as $\Lsym =
D^{-1/2}LD^{-1/2}$. 
Consider the natural random walk on $G$. Its transition matrix is
given as $P = D^{-1}W$.   It is well-known that $\lambda$ is an eigenvalue
of $\lsym$ if and only if $1-\lambda$ is an eigenvalue of $P$. By $1 =
\lambda_1  \geq \lambda_2 \geq \hd \geq \lambda_n>-1$ we denote the
eigenvalues of $P$. The quantity 
$1 - \max\{ \lambda_2, |\lambda_{n}|\}$ is called the {\em spectral
  gap} of $P$. \\

The {\em hitting time} $H_{uv}$ is defined as the expected
time it takes a random walk starting in vertex $u$ to travel to vertex
$v$ (where $H_{uu} = 0$ by definition).  The {\em commute distance}
({\em commute time}) between
$u$ and $v$ is defined as $C_{uv}  :=H_{uv} + H_{vu}$. 
Recall that for a symmetric, non-invertible 
matrix $A$ its Moore-Penrose inverse is 
defined as $A\pinv := (A + U)\inv - U$ where $U$ is the projection on
the eigenspace corresponding to eigenvalue 0. 
It is well known that commute times can be expressed in
terms of the Moore-Penrose inverse $L\pinv$ of the unnormalized graph
Laplacian (e.g., \citealp{KleRan93,XiaGut03,FouPirRenSae06}): 
\ba
R_{ij} = \inner{ e_i - e_j, L\pinv (e_i - e_j) } , 
\ea
where $e_i$ is the $i$-th unit vector in $\R^n$. 
The following representations for commute and hitting times
involving the pseudo-inverse $\lsym\pinv$ of the normalized
graph Laplacian are less well known. 

\begin{proposition}[Closed form expression for hitting and commute times]
\label{prop-commute} Let $G$ be a connected, undirected 
graph with $n$ vertices. 
The hitting times $H_{ij}$, $i \neq j$,   can be computed by 
\ba
H_{ij} = \vol(G) \Big\langle 
\frac{1}{\sqrt{d_j}} e_j \; , \;
\lsym\pinv
\Big( \frac{1}{\sqrt{d_j}}  e_j
-
\frac{1}{\sqrt{d_i}} e_i
\Big) 
\Big\rangle, 
\ea
and the commute times satisfy 
\ba
C_{ij} = \vol(G) \Big\langle 
\frac{1}{\sqrt{d_i}} e_i  - \frac{1}{\sqrt{d_j}} e_j \; , \;
\lsym\pinv
\Big( \frac{1}{\sqrt{d_j}}  e_j
-
\frac{1}{\sqrt{d_i}} e_i
\Big) 
\Big\rangle. 
\ea
\end{proposition}

Closely related to the 
commute distance is the {\em resistance distance}. Here one interprets
the graph as an electrical network where the edges represent resistors. 
The conductance of a resistor is given by the corresponding edge weight. 
The resistance distance $R_{uv}$ between two vertices $u$ and $v$ is defined as the effective
resistance between $u$ and $v$ in the network. It is well known that
the resistance distance coincides with the commute distance up to a
constant: $C_{uv}= \vol(G) R_{uv}$. 
For background reading on resistance and commute distances see
\citet{DoySne84,KleRan93,XiaGut03, FouPirRenSae06}. \\

Our main focus in this paper is the class of {\em geometric
  graphs}. For a {\em deterministic} (fixed) geometric graph we
consider a fixed set of points $X_1, \hd, X_n \in \R^d$. These points form the
vertices $v_1,\hd, v_n$ of the graph. The edges in the graph are defined such that
``neighboring points'' are connected. We consider the most
popular types of random geometric graphs. 
In the {\em $\eps$-graph} we connect two
points whenever their Euclidean distance is less than or equal to
$\eps$. In the undirected, {\em symmetric $k$-nearest neighbor graph} we connect
$v_i$ to $v_j$ if $X_i$ is among the $k$ nearest neighbors of $X_j$
{\em or} 
vice versa. In the {\em mutual $k$-nearest neighbor graph} we connect
$v_i$ to $v_j$ if $X_i$ is among the $k$ nearest neighbors of $X_j$
{\em and} 
vice versa. Note that by default, the terms $\eps$- and $\knn$-graph
refer to unweighted graphs in our paper. When we treat weighted
graphs, we always make it explicit. 
For a general {\em similarity graph} we build a weight matrix between
all points based on a similarity function $k: \R^d \times \R^d \to
\R_{\geq 0}$, that is we define the weight matrix $W$ with entries
$w_{ij} = k(X_i, X_j)$ and consider the fully connected graph with
weight matrix $W$. The most popular weight function in
applications is the Gaussian similarity function $w_{ij} = \exp(-
\|X_i - X_j\|^2 / h^2)$, where $h >0$ is a bandwidth parameter. \\ 

While these definitions make sense with any fixed set of vertices, we
are most interested in the case of {\em random} geometric graphs. Here
we assume that the underlying set of vertices $X_1, ..., X_n$ has been
drawn i.i.d. according to some probability density $p$ on $\R^d$. Once
the vertices are known, the edges in the graphs are constructed as
described above. In the random setting it is convenient to make
regularity assumptions in order to be able to control quantities such
as the minimal and maximal degrees. Sometimes we need to make these
assumptions about the whole underlying space, sometimes just for a
selected subset of it. Thus we introduce the following general
definition.

\begin{definition}[Valid region]\label{def-valid-region} Let $p$ be
  any density on $\R^d$. We call a connected subset $\Xcal
  \subset \R^d$ a
  {\em valid region} if the following
  properties are satisfied: 

\be 

\item The density on $\Xcal$ is bounded away from 0, that is for all
  $x \in \Xcal$ we have that $p(x) \geq p_{\min} > 0$ for some constant
  $p_{\min}$. %

\item $\Xcal$ has ``bottleneck'' larger than some
 value $h > 0$: the set $\{x\in\mathcal{X} \;:\; dist(x,\partial
 \mathcal{X}) > h/2\}$ is connected (here $\partial \Xcal$ denotes the
 topological boundary of $\Xcal$).

\item The boundary of $\Xcal$ is regular in the following sense. We
  assume that there exist positive constants $\alpha > 0$ and
  $\eps_0>0$ such that if $\eps < \eps_0$, then for all points 
$x \in \partial \Xcal$ we have $\vol( B_\eps(x) \cap \Xcal) \geq
  \alpha \vol(B_\eps(x))$ (where $\vol$ denotes the Lebesgue volume). Essentially this condition just excludes
  the situation where the boundary has arbitrarily thin spikes.  \ee

Sometimes we consider a {\em valid region with respect to two
points $s$, $t$}. Here we additionally assume that $s$ and $t$ are interior points of $\Xcal$. 
\end{definition}

In the spectral part of our paper, we always have to make a couple of
assumptions that will be summarized by the term  {\bf general
  assumptions}. They are as follows: 
First we assume that $\Xcal := \supp(p)$ is a valid region according to
Definition \ref{def-valid-region}. Second, we assume that $\Xcal$ does not contain any
holes and does not become arbitrarily narrow: there exists a
homeomorphism $ h: \Xcal \to [0,1]^d$ and constants $0 < L_{\min} < L_{\max} <
\infty$ 
such that for all $x, y \in \Xcal$ we have
\ba
L_{\min} \| x-y \| \; \leq \; 
\|h(x) - h(y)\| 
\;\leq\;
L_{\max} \| x-y \|. 
\ea
This condition restricts $\Xcal$ to be topologically 
equivalent to the cube. In applications this is not a strong
assumption, as the occurrence of ``holes'' with vanishing probability
density is unrealistic due to the presence of noise in the data
generating process. More generally we believe 
that our results
can be generalized to other homeomorphism classes, but refrain from
doing so as it would substantially increase the amount of
technicalities.\\

In the following we denote the volume of the unit ball in $\R^d$ by $\eta_d$. %
For readability reasons, we are going to state our main results using constants
$c_i > 0$. These constants are independent of $n$ and the graph
connectivity parameter ($\eps$ or $k$ or $h$, respectively) but depend on the
dimension, the geometry of $\Xcal$, and $p$. The values of all
constants are determined explicitly in the proofs. They are not the
same in different propositions. %
\section{Main results} \label{sec-main}

Our paper comprises two different approaches. In the first approach 
we analyze the resistance distance by flow based
arguments. This technique is somewhat restrictive in the sense that it
only works for the resistance distance itself (not the hitting times)
and we only apply it to random geometric graphs. The advantage is that in
this setting we obtain good convergence conditions and rates. 
The second approach is based on spectral arguments and is more
general. It works for various kinds of graphs and can treat hitting
times as well. This comes at the price of slightly stronger
assumptions and worse convergence rates. \\

\subsection{Results based on flow
  arguments} 

\begin{theorem}[Commute distance on $\eps$-graphs] \label{th-main-eps}
  Let $\Xcal$ be a valid region with bottleneck $h$ and minimal
  density $\pmin$. %
For $\eps \leq h$, consider an unweighted $\eps$-graph built from  the
sequence $X_1, \hd, X_n$ that has been drawn i.i.d. from the density
$p$. 
Fix $i$ and $j$. Assume that $X_i$ and $X_j$ have distance at
  least $h$ from the boundary of $\Xcal$, and that the distance
  between $X_i$ and $X_j$ is at least $8 \eps$. 
Then there exist
  constants $c_1, \hd, c_7 > 0$ (depending on the dimension and
  geometry of $\Xcal$) such that with probability at least $1
  - c_1 n \exp(-c_2 n \eps^d) - c_3 \exp(-c_4 n \eps^d) / \eps^d$ the
  commute distance on the $\eps$-graph satisfies 
  \ba %
  \left|\frac{ n \eps^d}{\vol(G)} C_{ij} - \left(\frac{n \eps^d}{d_i} + \frac{n
        \eps^d}{d_j}\right) \right| & \leq
\begin{cases} 
 c_5  / { n \eps^d}
& \hspace{1cm} \text{ if } d > 3\\
c_6 \cdot {\log(1/\eps) } / {n \eps^3 }
& \hspace{1cm} \text{ if } d = 3\\
c_7 / {n \eps^3 }
& \hspace{1cm} \text{ if } d = 2\\\end{cases}
\ea
The probability converges to 1 if $n \to \infty$ and $n \eps^d / \log(n) \to \infty$. The
right hand side  of the deviation bound converges to 0 as $n\to\infty$, if 
\ba
\begin{cases} 
n \eps^d \to \infty  &  \text{ if } d > 3\\
n \eps^3 / \log(1/\eps) \to \infty & \text{ if } d = 3\\
n \eps^3 = n \eps^{d+1} \to \infty & \text{ if } d = 2. 
\end{cases}
\ea
Under these conditions, if the density $p$ is continuous and if $\eps\to 0$, then 
\ba
\frac{ n \eps^d}{\vol(G)} C_{ij}  \to \frac{1}{\eta_d p(X_i)} + \frac{1}{\eta_d p(X_j)}
\;\;\;\;\text{ a.s.} 
\ea
\end{theorem}

\begin{theorem}[Commute distance  on
  $\knn$-graphs] \label{th-main-knn} Let $\Xcal$ be a valid region with bottleneck $h$ and density
  bounds $\pmin$ and $\pmax$. 
Consider an unweighted $\knn$-graph (either symmetric or mutual) such
that $(k/n)^{1/d}/2\pmax \leq h$, built from  the sequence $X_1, \hd,
X_n$ that has been drawn i.i.d. from the density $p$. 

Fix $i$ and $j$. Assume that $X_i$ and $X_j$ have distance at
  least $h$ from the boundary of $\Xcal$, and that the distance
  between $X_i$ and $X_j$ is at least $4 (k/n)^{1/d}/\pmax$. 
Then
  there 
exist constants 
$c_1, \hd, c_5 > 0$ 
such that with probability at least 
$1 - c_1  n \exp(-c_2 k)$ 
the commute distance on both the symmetric and the mutual $\knn$-graph
satisfies 
\ba %
\left|
\frac{k}{\vol(G)} C_{ij} - \left(\frac{k}{d_i} + \frac{k}{d_j} \right)\right|
& \leq  
\begin{cases} 
c_4 / {k } 
& \hspace{1cm} \text{ if } d > 3\\
c_5 \cdot { \log(n/k)} / {k} 
& \hspace{1cm} \text{ if } d = 3\\
c_6 {n^{1/2}} / {k^{3/2}}
& \hspace{1cm} \text{ if } d = 2\\
\end{cases}
\ea
The probability converges to 1 if $n \to \infty$ and $k
/ \log(n) \to \infty$. In case $d > 3$, the right hand side  of the deviation bound
converges to 0 if $k \to \infty$ (and under slightly worse conditions
in cases $d=3$ and $d=2$). 
Under these conditions, if the density $p$ is
continuous and if additionally $k/n\to 0$,
then $
\frac{k}{\vol(G)} C_{ij} \to 2$ almost surely. 
\end{theorem}

Let us make a couple of technical remarks about these theorems. \\

To achieve the convergence of the commute distance we
have to rescale it appropriately (for example, in the $\eps$-graph we
scale by a factor of $n \eps^d$). Our rescaling is exactly chosen such
that the limit expressions are finite, positive values. 
Scaling by any other factor in terms of $n$, $\eps$ or $k$ either
leads to divergence or to convergence to zero. \\

In case $d>3$, all convergence conditions on $n$ and $\eps$ (or $k$, respectively)
are the ones to be expected for random geometric graphs. They are
satisfied as soon as the degrees grow faster than $\log(n)$ (for
degrees of order smaller than $\log(n)$, the graphs are not connected anyway, see
e.g. \citealp{Penrose99}). Hence, our results hold for sparse as well
as for dense connected random geometric graphs. In dimensions 3 and 2,
our rates are not of the same flavor as in the higher dimensions. For example, in dimension
2 we need $n \eps^3 \to \infty$ instead of $n \eps^2 \to
\infty$. On the one hand we are not too surprised to get systematic
differences between the lowest few dimensions. The same happens in
many situations, just consider the example of Polya's
theorem about the recurrence/ transience of random walks on grids. On
the other hand, these differences might as well be an artifact of our
proof methods (and we suspect so at least for the case $d=3$; but even
though we tried, we did not get rid of the log factor in this
case). It is a matter of future work to clarify this. \\

The valid region $\Xcal$ has been introduced for technical reasons. We need to
operate in such a region in order to be able to control the behavior
of the graph, e.g. the minimal and maximal degrees. The assumptions on $\Xcal$ are
the standard assumptions used regularly in the random geometric graph
literature. In our setting, we have the freedom of choosing $\Xcal
\subset \R^d$ as we want. In order to
obtain the tightest bounds one should aim for a valid $\Xcal$ that has 
a wide bottleneck $h$ and a high minimal density $\pmin$. 
In general this freedom of choosing $\Xcal$ shows that if two points are in the same
high-density region of the space, the convergence of the commute
distance is very fast, while it gets slower if the two points are in
different regions of high density separated by a bottleneck.\\

We stated the theorems above for a fixed pair $i,j$. However, they
also hold uniformly over all pairs $i,j$ that satisfy the conditions
in the theorem (with exactly the same statement). The reason is that
the main probabilistic quantities that enter the proofs are bound on the minimal and
maximal degrees, which of course hold uniformly.

\subsection{Results based on spectral arguments} \label{sec-main-spectral}

The representation of the hitting and commute
times in terms of the Moore-Penrose inverse of the normalized graph
Laplacian (Proposition \ref{prop-commute}) can be used to derive the following key proposition that is the
basis for all further results in this section.  

\begin{proposition}[Absolute and relative bounds in any fixed graph]
\label{prop-approx} \sloppy

Let $G$ be a finite, connected, undirected, possibly weighted graph
that is not bipartite.

\be 
\item For $i\neq j$
\banum \label{eq-approx-hitting}
\left|
\frac{1}{\vol(G)} H_{ij} 
- \frac{1}{d_j}
\right|
\;\;\leq \;\;
2 \left( \frac{1}{1 - \lambda_2} + 1\right) 
\frac{w_{\max}}{d_{\min}^2}.
\eanum

\item For $i\neq j$
\banum \label{eq-approx-commute}
\left|
\frac{1}{\vol(G)} C_{ij} 
- \bigg(\frac{1}{d_i} + \frac{1}{d_j}\bigg)
\right|
\;\;\leq \;\;
\frac{\wmax}{\dmin} \bigg(\frac{1}{1-\lambda_2} +
2\bigg)\bigg(\frac{1}{d_i} + \frac{1}{d_j}\bigg)
\;\;\leq \;\;
2 \left( \frac{1}{1 - \lambda_2} + 2\right) 
\frac{w_{\max}}{d_{\min}^2}.
\eanum
\ee
\end{proposition}

\vspace{1cm}

We would like to point out that even though the bound in Part 3 of the
proposition is reminiscent
to statements in the literature, it is much tighter. Consider the
following formula from \citet{Lovasz93}
\ba
\frac{1}{2} \left(\frac{1}{d_i} + \frac{1}{d_j} \right)
\;\; \leq \;\; \frac{1}{\vol(G)}C_{ij} \;\; \leq \;\; 
\frac{1}{1 - \lambda_2} \left(\frac{1}{d_i} + \frac{1}{d_j} \right) 
\ea
that can easily be rearranged to the following bound: 
\banum \label{eq-lovasz}
\left| \frac{1}{\vol(G)}C_{ij} - \bigg(\frac{1}{d_i} + \frac{1}{d_j}\bigg) \right|
\leq 
\frac{1}{1 - \lambda_2}  \frac{2}{\dmin}.
\eanum
The major difference between our bound \eqref{eq-approx-commute} and Lovasz'
bound \eqref{eq-lovasz} is that while the latter has the term $\dmin$ in the denominator, our bound
has the term $\dmin^2$ in the denominator. This makes all of a
difference: in the graphs under considerations our
bound converges to 0 whereas Lovasz'
bound diverges.

\subsubsection{Application to unweighted random geometric graphs}

In the following we are going to apply Proposition \ref{prop-approx} to various
random geometric graphs. Next to some standard results
about the degrees and number of edges in random geometric graphs, the
main ingredients are the following
bounds on the spectral gap in random geometric graphs. These bounds
are of independent interest because the spectral gap governs many important
properties and processes on graphs. \\

\begin{theorem}[Spectral gap of the $\eps$-graph]
\label{th-gap-eps}
Suppose that the general assumptions hold. 
Then there exist constants
$c_1, \hd, c_6 > 0$ such that with probability at
least $1 - c_1 n \exp( -c_2 n \eps^d)  - c_3 \exp(- c_4 n \eps^d) /
\eps^d$ 
\ba
1 - \lambda_2 \;\geq\; c_5 \cdot \eps^2 && 
\text{ and } && 
1 - |\lambda_{n}| \;\geq\; c_6  \cdot \eps^{d+1} / n. 
\ea
If $n \eps^d / \log n \to \infty$, then this probability converges to
1. 
\end{theorem}
\begin{theorem}[Spectral gap of the $\knn$-graph]
\label{th-gap-knn}
Suppose that the general assumptions hold. Then for both the symmetric
and the mutual $\knn$-graph there exist constants
$c_1, \hd, c_4 >0$ such that with probability at
least $1 - c_1 n \exp( - c_2 k )$, 
\ba
1 - \lambda_2 \;\geq \;c_3 \cdot (k/n)^{2/d} 
&& \text{ and } &&
1 - |\lambda_{n}| 
\;\geq \;c_4 \cdot k^{2/d} / n^{ (d+2)/d}
\ea
If $k / \log n \to \infty$, then the probability converges to 1. 
\end{theorem}

At first glance it seems surprising that the geometry of the
underlying space $\Xcal$ does not affect the order of magnitude of the
spectral gap, these quantities only enter the bound in terms of the
constants (as can be seen in the proofs below).  In particular, for large
$n$ the spectral gap does not depend on whether $\Xcal$ has a
``bottleneck'' or not. Intuitively this is the case because if the
sample size is large, even a bottleneck with very small
diameter contains many sample points and ``appears wide'' to the random walk. \\

The following theorems characterize the
hitting and commute times for $\eps$-and $\knn$-graphs. They are
direct consequences of plugging the results about the spectral gap
into Proposition \ref{prop-approx}. 

\begin{corollary}[Hitting and commute times on $\eps$-graphs]
\label{th-hitting-eps}
Assume that the general assumptions hold. 
Consider an unweighted $\eps$-graph built from  the sequence $X_1,
\hd, X_n$ drawn i.i.d. from the density $p$. 
Then there exist constants
$c_1, \hd, c_5 > 0$ such that with probability at
least $1 - c_1 n \exp(- c_2 n \eps^d) -c_3 \exp( - c_4 n \eps^d) /
\eps^d$, we have uniformly for all $i \neq j$ that 
\banum %
& \left|
\frac{n \eps^d}{\vol(G)} H_{uv}
- 
\frac{n \eps^d}{d_v} 
\right| \;\; 
\leq \;\; 
\frac{c_5}{n \eps^{d+2}}. 
\eanum
If  %
the density $p$ is continuous
and $n
\to \infty, \eps \to 0$ and $n \eps^{d+2}
\to \infty$, then 
\ba
\frac{n \eps^d}{\vol(G)} H_{ij} 
\longrightarrow 
\frac{1}{\eta_d \cdot p(X_j)} \;\;\text{almost surely.}
\ea
For the commute times, the analogous results hold due to $C_{ij} = H_{ij} + H_{ji}$. 
\end{corollary}
\begin{corollary}[Hitting and commute times on $\knn$-graphs]
\label{th-hitting-knn}
 Assume that the general assumptions hold. 
Consider an unweighted $\knn$-graph built from  the sequence $X_1, \hd, X_n$ drawn i.i.d. from the density $p$. 
Then for both the symmetric
 and mutual $\knn$-graph there exist constants
 $c_1, c_2, c_3 > 0$ such that  with probability at
 least \sloppy \mbox{$1 - c_1 \cdot n \cdot \exp(- k c_2)$},  we have uniformly for all $i \neq j$ that 
\banum \label{eq-dev-knn} 
& \left|
\frac{k}{\vol(G)} H_{ij}
- 
\frac{k}{d_j} 
\right| \;\; 
\leq \;\;
c_3 \cdot \frac{n^{2/d}}{k^{1 + 2/d}}. 
\eanum
If  
the density $p$ is
continuous and $n \to \infty$, ${k} / {n}\rightarrow 0$ and $
k\big(k / n\big)^{2/d} \to \infty$, then 
\ba
\frac{k}{\vol(G)} H_{ij} 
\longrightarrow 
1 \;\;\text{almost surely.}
\ea
For the commute times, the analogous results hold due to $C_{ij} = H_{ij} + H_{ji}$. 
\end{corollary}

\subsubsection{Application to weighted graphs}

In several applications, $\eps$-graphs or $\knn$ graphs are not used
as unweighted graphs, but additionally endowed with edge weights. For
example, in the field of machine learning it is common to use 
Gaussian weights  $w_{ij} = \exp(-
\|X_i - X_j\|^2 / h^2)$, where $h >0$ is a bandwidth parameter. \\

We can use standard spectral results to prove approximation theorems
in such cases. \\

\begin{theorem}[Results on fully connected weighted
  graphs]\label{th:Hitting-fully-connected} \sloppy
Consider a fixed, fully connected weighted graph with weight matrix
$W$. Assume that its entries are upper and lower bounded by some
constants $\wmin, \wmax$, that is $0 <
\wmin \leq w_{ij} \leq \wmax$ for all $i,j$.  Then,
uniformly for all $i,j \in \{1, ..., n\}$, $i \neq j$,  
\begin{align*}
\left| \frac{n}{\vol(G)} H_{ij}  - \frac{n}{d_j} \right| \;\leq \; 
4 n \left( \frac{w_{\max}}{w_{\min}} \right)
\frac{w_{\max}}{d_{\min}^2} 
\;\;\leq\;\;
 4 \frac{\wmax^2}{\wmin^3}\frac{1}{n}.
\end{align*}
\end{theorem}

For example, this result can be applied directly to a Gaussian
similarity graph (for fixed bandwidth $h$). \\

The next theorem treats the case of Gaussian similarity graphs with
adapted bandwidth $h$. The technique we use to prove  this theorem is very
general. Using the Rayleigh principle, we reduce the case of the fully connected Gaussian graph to
a truncated graph where edges beyond a certain length are
removed. Bounds for this truncated graph, in turn, can be reduced to
bounds  of the unweighted $\eps$-graph. 
With this technique it is possible to treat
very general classes of graphs. 

\begin{theorem}[Results on Gaussian graphs with
  adapted bandwidth] \label{co:Gaussian}
Let $\mathcal{X}\subseteq\mathbb{R}^d$ be a compact set and  $p$ a continuous, strictly positive density on $\Xcal$. 
Consider a fully connected, weighted similarity graph built from the points $X_1, \hd, X_n$ drawn i.i.d. from $\Pr$ with density $p$. 
As weight function use the Gaussian similarity function 
$k_{h}(x,y) =
\tfrac{1}{(2\pi h^2)^\frac{d}{2}}\exp\left(-\frac{\norm{x-y}^2}{2h^2}\right)$. 
If  %
the density $p$ is continuous
and $n \to \infty, h \to 0$ and $n h^{d+2}/\log(n) \to \infty$, then 
\ba
\frac{n }{\vol(G)} C_{ij} 
\longrightarrow 
\frac{1}{p(X_i)} + \frac{1}{p(X_j)} \;\;\text{almost surely.}
\ea
\end{theorem}

Note that in this theorem, we
introduced the scaling factor $1/h^d$ already in
the definition of the Gaussian similarity function to obtain the
correct density estimate $p(X_j)$ in the limit. For this reason, the resistance results
are rescaled with factor $n$ instead of $n h^d$.

\subsubsection{Application to random graphs with given expected
  degrees and \ER graphs}

Consider the general random
graph model where the edge between vertices $i$ and $j$ is chosen independently with
a certain probability $p_{ij}$ that is allowed to depend on $i$ and
$j$. This model contains very popular random graph models such as 
the \er random graph, planted partition graphs, and random
graphs with given expected degrees. 
For this class of random graphs, the following result has been proved
recently by \citet{ChuRad11}. 

\begin{theorem}[\citealp{ChuRad11}] 
\label{th-chungradcliffe}Let $G$ be a random graph where
  edges between vertices $i$ and $j$ are put independently with probabilities $p_{ij}$.Consider the normalized
  Laplacian $\Lsym$, and define the expected normalized Laplacian as
  the matrix $\overline{Lsym} : = I - \overline{D^{-1/2}}  \overline{A} 
    \overline{D^{-1/2}}$ where $\overline{A}_{ij} = E(A_{ij}) =
    p_{ij}$ and $\overline{D} = E(D)$.  Let $\overline{\dmin}$ be the
  minimal expected degree. Denote the eigenvalues of $\lsym$ by $\mu$, the ones of
  $\overline \lsym$ by $\overline\mu$. Choose $\eps >0$. Then there exists a constant $k =
  k(\eps)$ such that if $\overline\dmin > k \log(n)$, then with
  probability at least $1 - \eps$, 
\ba
\forall {j =1, ..., n}: \;\;\;| \mu_j - \overline{\mu_j} | \leq 2 \sqrt{\frac{3 \log(4n / \eps)}{\overline{\dmin}}}\\
\ea
\end{theorem}

{\em Application to \er graphs. }
Here all edges have constant probabilities $p_{ij} =
p$ (for simplicity, we also allow for self-edges). 

\begin{corollary}[Results  on \ER graphs] \label{th-er}
Let $n \to \infty$, $p = \omega( \log n / n)$. Then the rescaled
hitting times 
on the \ER graph converge to a constant: for all
vertices $u,v$ in the graph we have 
\ba
\left|\frac{1}{n} \cdot H_{uv} - 1 \right|  = O\left(\frac{1}{n
    p}\right) \to 0 \;\;\;\text{ in probability.} \\
\ea
\end{corollary}

{\em Application to planted partition graphs.\\} 
Next we consider a simple model of an \er-graph with planted
partitions, the planted bisection case. Assume that the $n$ vertices are split into two
``clusters'' of equal size. We put an edge between two vertices $u$ and $v$ with
probability $p_{within}$ if they are in the same cluster and with
probability $p_{between} < p_{within}$ if they are in different
clusters. For simplicity we allow self-loops.

\begin{corollary}[Random graph with planted partitions]
\label{th-er-planted}
Consider an \er graph with planted bisection. Assume that 
$p_{within} = \omega(\log(n)/n)$ and $p_{between}$ such that
$n p_{between} \to \infty$ (arbitrarily slow). Then, for all vertices $u,v$ in
the graph
\ba
\left| \frac{1}{n} 
\cdot H_{ij} -1 \right| = O\left(\frac{1}{n \;\pbetween}\right) \to 0 \;\;\;\text{ in probability}. \\
\ea
\end{corollary}

This result is a prime example to show how that even though there is a
strong cluster structure in the graph, the hitting times and commute
distances cannot see this cluster structure any more, once the graph
gets too large. Note that the corollary even holds if $p_{between}$ grows much
slower than $p_{within}$. That is, the larger our graph, the more
pronounced is the cluster structure. Nevertheless, the commute
distance converges to a trivial result. On the other hand, we also see
that the speed of convergence is $O(n \pbetween)$, that is, if
$\pbetween = g(n) /n$ with a very slow growing function $g$, then
convergence can be very slow. We might need
very large graphs before the degeneracy of the commute time will be
visible. \\

{\em Application to random graphs with given expected degrees.} 
For a graph of $n$ vertices we have $n$
parameters $\bar d_1,
..., \bar d_n > 0$. For each pair of vertices $v_i$ and $v_j$, we independently place an
edge between these two vertices with probability $\bar d_i \bar d_j /
\sum_{k=1}^n \bar d_k$. It is easy to see that in this model, vertex
$v_i$ has expected degree $\bar d_i$ (cf.\ Section 5.3. in
\citealp{ChuLu06} for background reading). \\

\begin{corollary}[Results on random graphs with expected degrees] \label{th-expected-degrees}
Consider any sequence of random graphs with expected degrees such that
$\overline{d_{\min}} = \omega(\log n)$. Then the commute distances satisfy for all $i \neq j$,
\ba
\frac{ \bigg|\frac{1}{\vol(G)} C_{ij} -
  \bigg(\frac{1}{d_i}+\frac{1}{d_j}\bigg)\bigg|}{\frac{1}{d_i}+\frac{1}{d_j}} 
= O\left(\frac{1}{\log(2n)}\right) \longrightarrow 0,  \text{ almost surely.} \\
\ea
\end{corollary}

\section{Proofs for the flow-based approach} \label{sec-proofs-flow}

For notational convenience, in this section we work with the
resistance distance $R_{uv} = C_{uv} / \vol(G)$ instead of
the commute distance $C_{uv}$, then we do not have to carry the factor
$1/\vol(G)$ everywhere. 

\subsection{Lower bound}

It is easy to prove that the resistance distance between two points is
lower bounded by the sum of the inverse degrees. 

\begin{proposition}[Lower bound] \label{th-lower}
Let $G$ be a weighted, undirected, connected graph and consider two
vertices $s$ and $t$, $s\neq t$. Assume that $G$  remains connected if we
remove $s$ and $t$. Then the effective resistance between $s$ and $t$
is bounded by 
\ba
R_{st} \geq \frac{ Q_{st}}{1 + w_{st} Q_{st} }
\ea
where $Q_{st} = 1/(d_s - w_{st}) + 1/(d_t - w_{st})$. Note that if $s$
and $t$ are not connected by a direct edge (that is, $w_{st}=0$), then
the right hand side simplifies to
$1/d_s + 1/d_t$. 
\end{proposition}

\begin{proof}
  The proof is based on Rayleigh's monotonicity principle that states
  that increasing edge weights in the graph can never increase the
  effective resistance between two vertices (cf. Corollary 7 in
  Section IX.2 of \citealp{Bollobas98}). Given our original graph $G$,
  we build a new graph $G'$ by setting the weight of all edges to
  infinity, except the edges that are adjacent to $s$ or $t$ (setting
  the weight of an edge to infinity means that this edge has no
  resistance any more). This can
  also be interpreted as taking all vertices except $s$ and $t$ and
  merging them to one super-node $a$. Now our graph $G'$ consists of
  three vertices $s, a, t$ with several parallel edges from $s$ to
  $a$, several parallel edges from $a$ to $t$, and potentially the
  original edge between $s$ and $t$ (if it existed in $G$). 
  Exploiting
  the laws in electrical networks (resistances add along edges in
  series, conductances add along edges in parallel; see Section 2.3 in
  \citet{LyoPer10} for detailed instructions and examples) leads to
  the desired result.
\end{proof}

\subsection{Upper bound}

This is the part that requires the hard work. 
Our proof is based on a theorem that shows how the resistance between
two points in the graph can be computed in terms of flows on the
graph. The following result is taken from Corollary 6 in Section IX.2 of
\citet{Bollobas98}.

\begin{theorem}[Resistance in terms of flows,
  cf. \citealp{Bollobas98}]
\label{th-bollobas}
Let $G=(V,E)$ be a weighted graph with edge weights $w_e$ $(e \in
E)$. 
The effective resistance $R_{st}$ between two fixed vertices $s$ and
$t$ can
be expressed as 

\banum \label{eq-bollobas}
R_{st} = 
\inf \left\{ 
\sum_{e \in E} \frac{u_e^2 }{w_e}
\; \Big| \;
u = (u_e)_{e \in E} \text{ unit flow from } s \text{ to } t
\right\}.
\eanum

\end{theorem}

Note that evaluating the formula in the above theorem for any fixed
flow leads to an upper bound on the effective resistance. The key to
obtaining a tight bound is to distribute the flow as widely and
uniformly over the graph as possible. \\

For the case of geometric graphs we are going to use a grid on
the underlying space to construct an efficient flow between two
vertices.  
Let $X_1, ..., X_n$ be a fixed set of points in $\R^d$ and consider
a geometric graph $G$ with vertices $X_1, ..., X_n$. Fix any two
of them, say $s:= X_1$ and $t: = X_2$. Let $\Xcal \subset \R^d$ be a connected set
that contains both $s$ and $t$. Consider a  regular grid with grid
width $g$ on $\Xcal$. We say that grid cells are neighbors of each other if they touch
  each other in at least one edge. 

\begin{definition}[Valid grid]\label{def-valid-grid}
We call the grid {\em valid } if the following properties
are satisfied: 
\be
\item The grid width is not too small: Each cell of the grid contains at least one of the points $X_1,
  ..., X_n$. 
\item The grid width $g$ is not too large: Points in the same or neighboring cells of the grid are always
  connected in the graph $G$. 

\item Relation between grid width and geometry of $\Xcal$: Define the bottleneck $h$ of
  the region $\Xcal$ as the
  largest $u$ such that the set $\{x \in \Xcal \condon
  dist(x, \partial \Xcal) > u/2\} $ is connected. 

We require that
  $\sqrt{d}\;g \leq h$ (a cube of side length $g$ should fit in the
  bottleneck).  \ee
\end{definition} 

We now prove the following
general  proposition that gives an upper bound on the resistance
distance between vertices in a fixed geometric graph.

\begin{proposition}[Resistance on a fixed geometric graph]
\label{prop-resistance-fixed-graph}
\sloppy 
Consider a fixed set of points 
$X_1, ..., X_n$ in some connected region $\Xcal \subset \R^d$ 
and a geometric graph on $X_1,..., X_n$. 
Assume that $\Xcal$ has bottleneck not smaller than $h$ (where the bottleneck is defined as in the definition
of a valid grid). Denote
$s=X_1$ and $t=X_2$. Assume that $s$ and $t$ can be connected by a
straight line that stays inside $\Xcal$ and has distance at least
$h/2$ to $\partial \Xcal$. Denote the distance between $s$ and $t$ by
$d(s,t)$. Let  $g$ be the width of a valid
grid on $\Xcal$ and assume that $d(s,t) > 4\sqrt{d}\;g$. By $\nmin$ denote the minimal number of points in
each grid cell, and define $a$ as 
\banum \label{eq-a}
 a \;\;  := \;\;\left\lfloor h/ (2
  g\sqrt{d-1})\right\rfloor . 
\eanum 
Assume that points that are connected in the graph are at most $Q$
grid cells apart from each other (for example, two points in the two grey cells in
Figure \ref{fig-step2} are 5 cells apart from each other). 
Then the effective 
resistance between $s$ and $t$ can be bounded as follows: 

\banum \label{eq-approx-resistance}
& \text{In case } d > 3: &&
R_{st} 
\leq 
\frac{1}{d_s} + \frac{1}{d_t} + 
\left(
\frac{1}{d_s} + \frac{1}{d_t}
\right)
\frac{2}{\nmin}
+ 
\frac{1}{\nmin^2} 
\left(3   + \frac{d(s,t)}{g(2a+1)^3 } + 2Q\right)\\
& \text{In case } d = 3: &&
R_{st} 
\leq 
\frac{1}{d_s} + \frac{1}{d_t} + 
\left(
\frac{1}{d_s} + \frac{1}{d_t}
\right)
\frac{2}{\nmin}
+ 
\frac{1}{\nmin^2} 
\left(\log(a) + 2 + \frac{d(s,t)}{g(2a+1)^2 } + 2Q\right)\\
& \text{In case } d = 2: &&
R_{st} 
\leq 
\frac{1}{d_s} + \frac{1}{d_t} + 
\left(
\frac{1}{d_s} + \frac{1}{d_t}
\right)
\frac{2}{\nmin}
+ 
\frac{1}{\nmin^2} 
\left(2a + 2  + \frac{d(s,t)}{g(2a+1) } +2Q\right)
\eanum
\end{proposition}

The general idea of the proof is to construct a flow from $s$ to $t$ with the help of
the underlying grid. On a high
level, the construction of the proof is not so difficult, but the
details are lengthy and a bit tedious. The  rest of this
section is devoted to it. \\

{\bf Construction of the flow --- overview.} 
Without loss of generality we assume that there exists a straight line
connecting $s$ and $t$ which is
along the first dimension of the space.

{\be 
\addtocounter{enumi}{-1}
\renewcommand{\labelenumi}{Step \arabic{enumi}: }

\item 
We start a unit flow in vertex $s$. 

\item We make a step to all neighbors $\neigh(s)$ of $s$ and distribute the
flow uniformly over all edges. That is, we traverse $d_s$ edges and send
flow $1/d_s$ over each edge (see Figure \ref{fig-step1}).

\item Some of the flow now sits inside $C(s)$, but some of it might
  sit outside of $C(s)$. In this step, we bring back all flow to
  $C(s)$ in order to control it later on (see Figure \ref{fig-step2}). \\ 

\item We now distribute the flow from $C(s)$ to a
larger region, namely to a hypercube $H(s)$ of side length $h$ that is
perpendicular to the linear path from $s$ to $t$ and centered at
$C(s)$ (see the hypercubes in Figure \ref{fig-overview}). This can be achieved in
several substeps that will be defined below. 

\item We now traverse from $H(s)$ to an analogous hypercube $H(t)$
located at $t$ using parallel paths, see Figure \ref{fig-overview}. 

\item From the hypercube $H(t)$ we send the flow to the neighborhood
  $\neigh(t)$ (this is the ``reverse'' of steps 2 and 3). 

\item  From $\neigh(t)$ we finally send the flow to the destination
$t$ (``reverse'' of step 1). \\

\ee
}

\begin{figure}
\bc
\subfloat[Step 1. Distributing the flow from $s$ (black dot) to all its
  neighbors (grey dots). %
]{%
\includegraphics[width=0.2\textwidth]{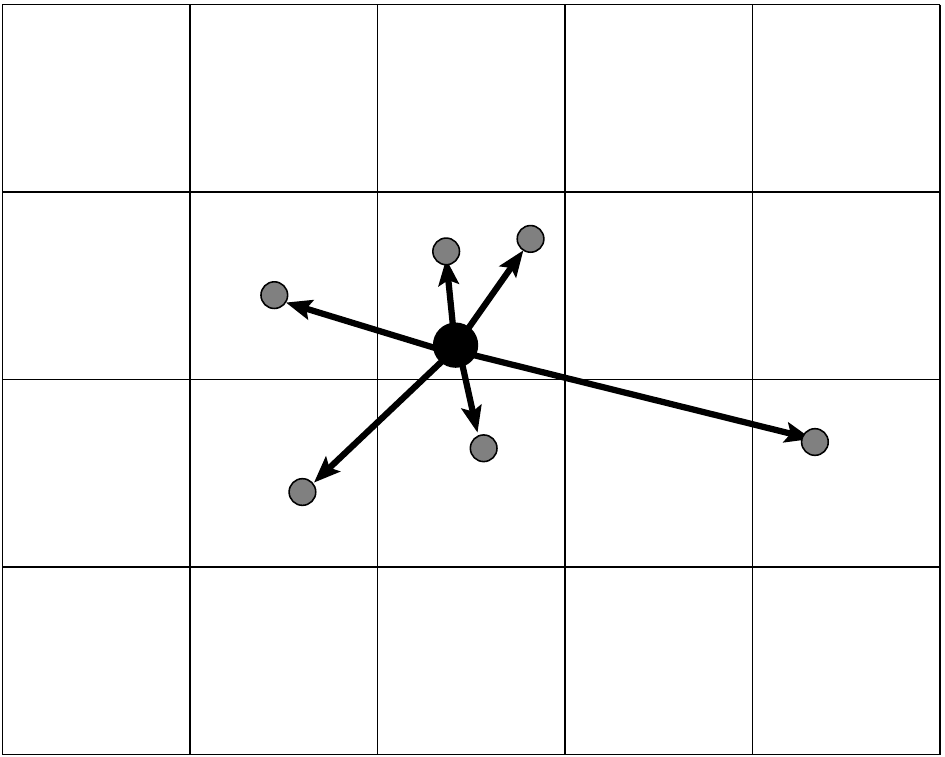} 
\label{fig-step1}}
\hspace{1cm}
\subfloat[Step 2. We bring back all flow from $p$ to $C(s)$. Also
shown in the figure is the hypercube to
  which the flow will be expanded in Step 3.]{%
\includegraphics[width=0.2\textwidth]{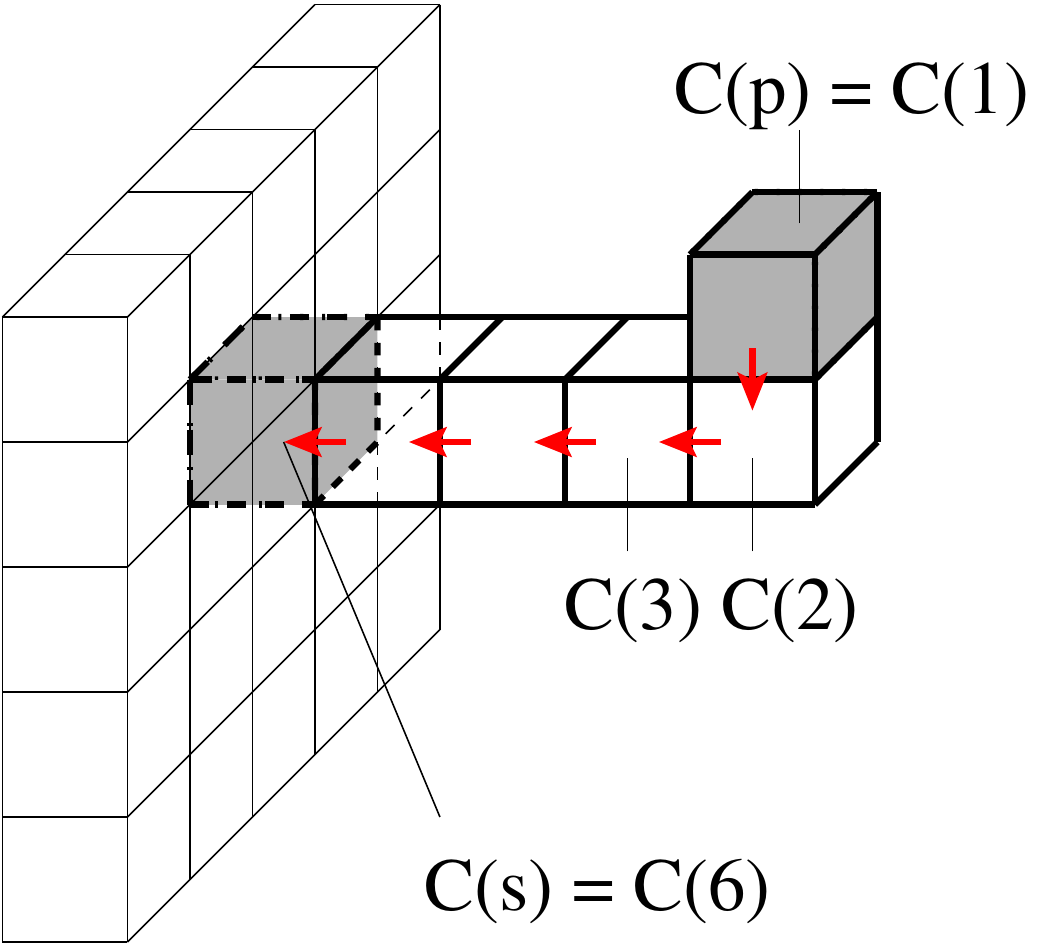} 
\label{fig-step2}}
\hspace{1cm}
\subfloat[Steps 3 and 4 of the flow construction: distribute the flow
  from $C(s)$ to a
  ``hypercube'' $H(s)$, then transmit it to a similar hypercube $H(t)$
  and guide it to $C(t)$.]{\includegraphics[width=0.4\textwidth]{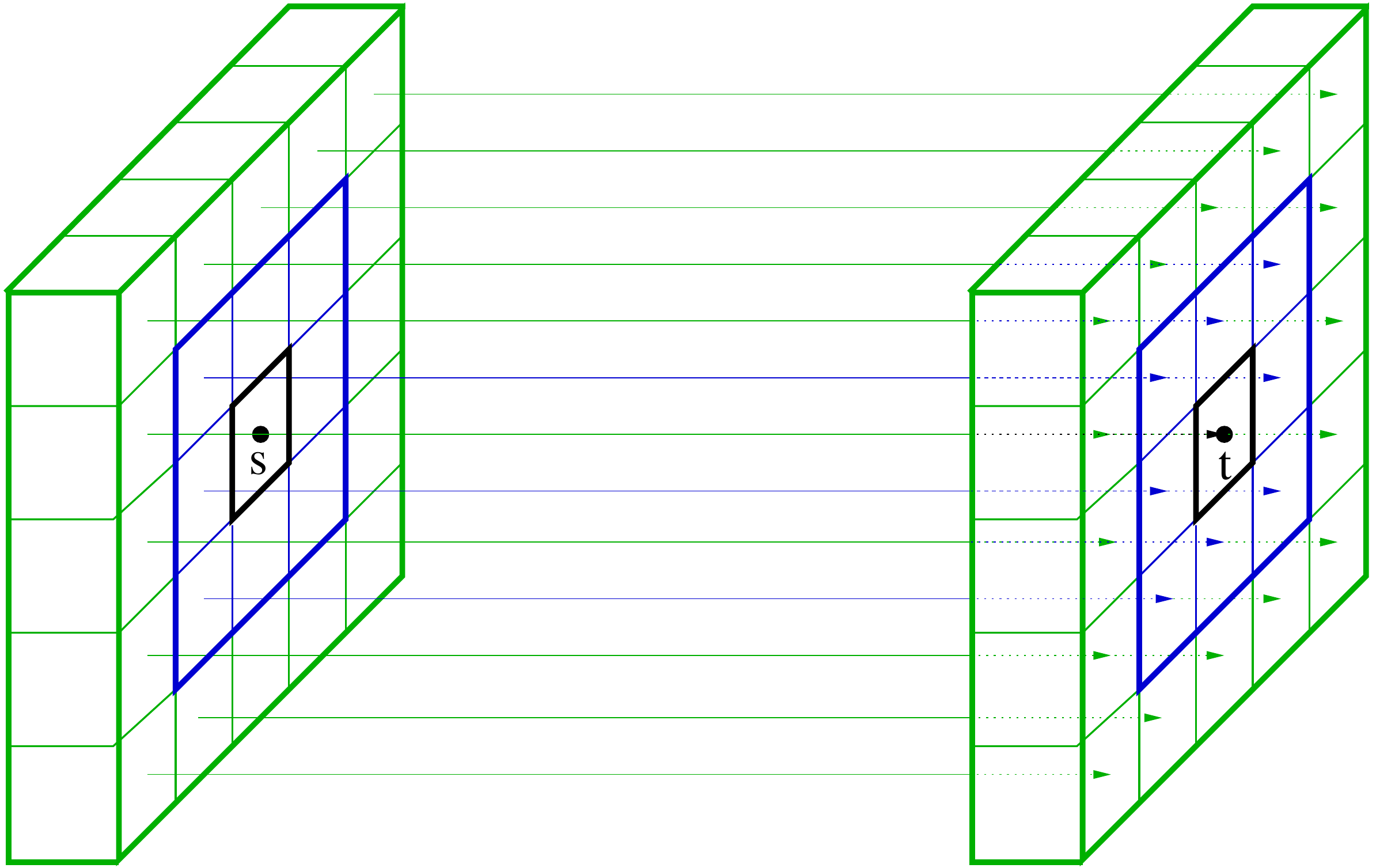}
\label{fig-overview}}
\ec
\caption{The flow construction --- overview. }
\end{figure}
\begin{figure}
\subfloat[Definition of layers.]{%
\includegraphics[width=0.2\textwidth]{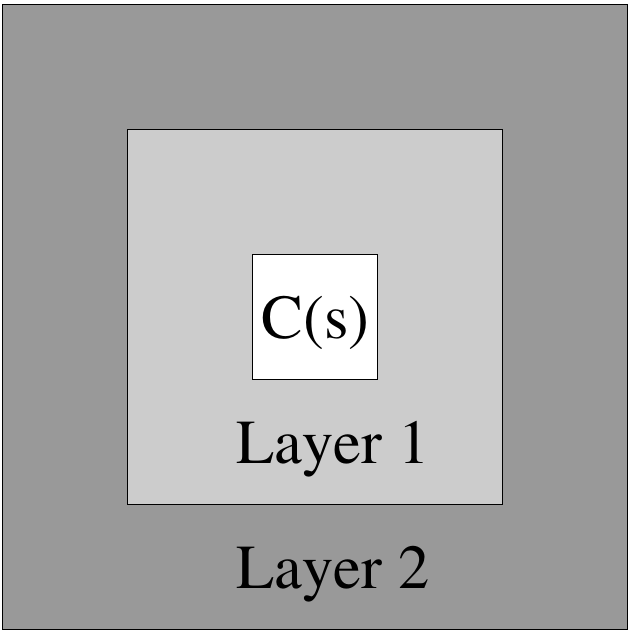}
\label{fig-layers}}
\hspace{0cm}
\subfloat[Before Step 3a starts, all flow is
  uniformly distributed in Layer $i-1$ (dark area). ]{%
\includegraphics[width=0.2\textwidth]{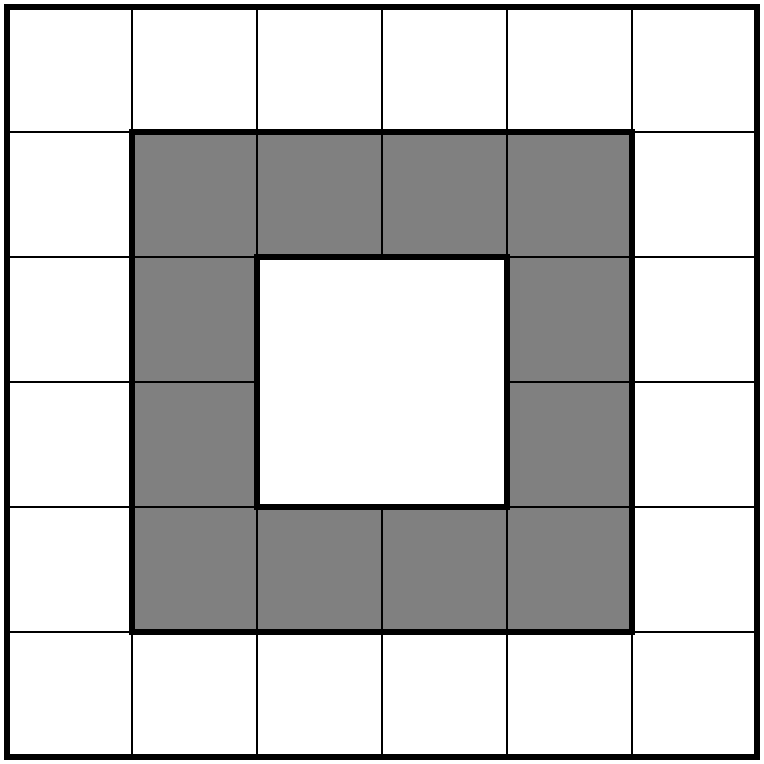}}
\hspace{0.3cm}
\subfloat[Step 3a then
  distributes the flow from Layer $i-1$ to the adjacent cells in Layer $i$]{%
\includegraphics[width=0.2\textwidth]{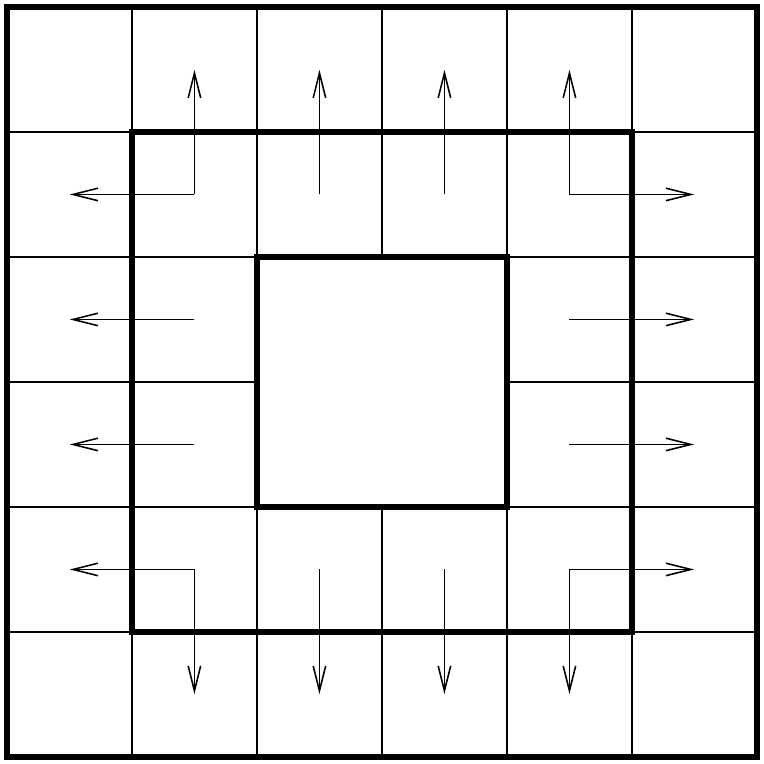}}
\hspace{0.3cm}
\subfloat[After Step 3a: all flow is in Layer $i$, but not yet
  uniformly distributed]{
\includegraphics[width=0.2\textwidth]{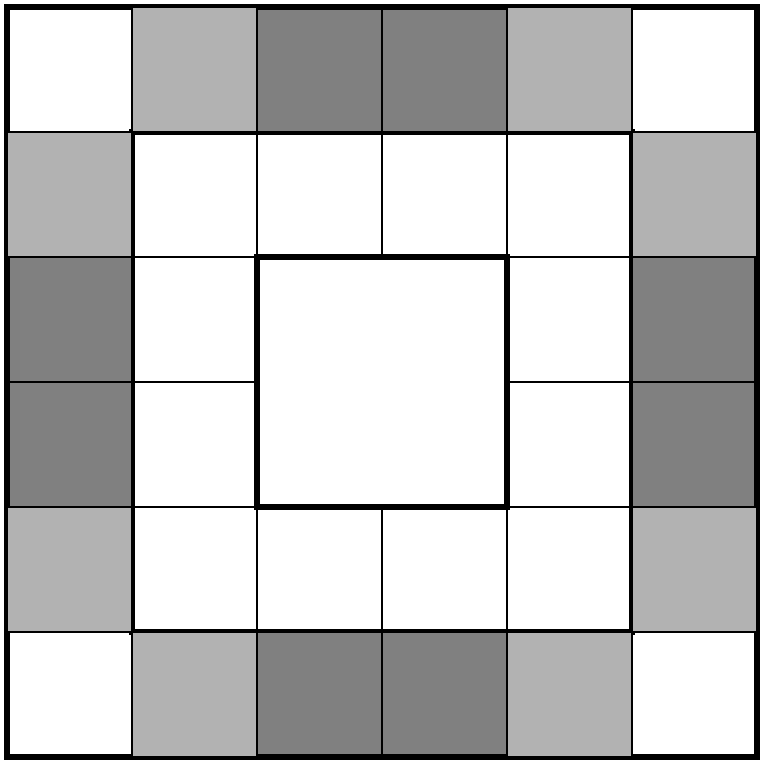}}
\hspace{0.3cm}
\subfloat[Step 3b
  redistributes the flow in Layer $i$.]{%
\includegraphics[width=0.2\textwidth]{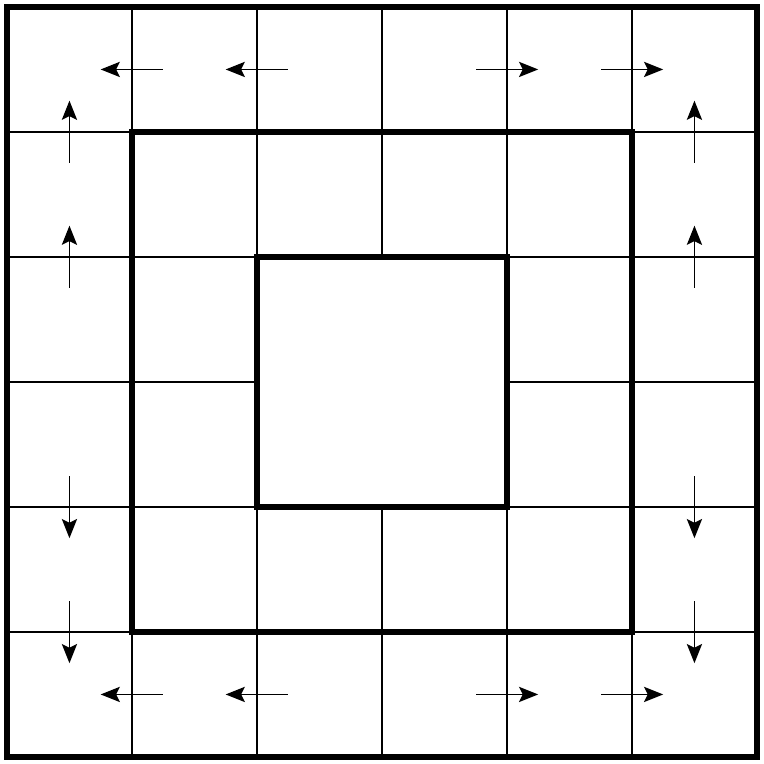}}
\hspace{0.3cm}
\subfloat[After Step 3b, the flow is uniformly distributed in Layer
$i$. ]{%
\includegraphics[width=0.2\textwidth]{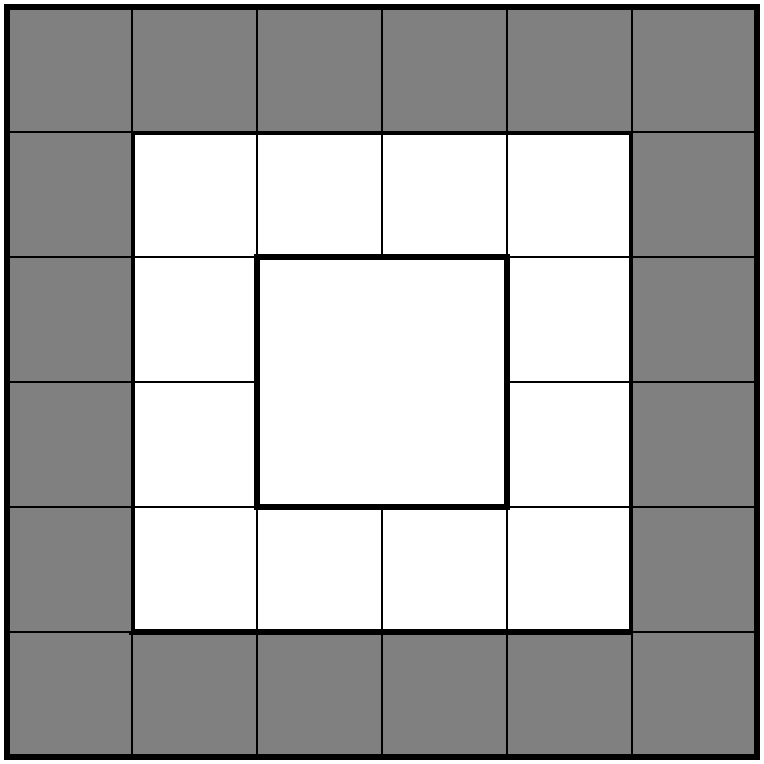}}
\caption{Details of Step 3 between Layers $i-1$ and $i$. The first row
  corresponds to the expansion phase, the second row to the
  redistribution phase. The figure is shown for the case of $d=3$. }
\label{fig-step3}
\end{figure}

{\bf Details of the flow construction and computation of the
  resistance beween $s$ and $t$ in the general case $d>3$. }
We now describe the individual steps and their contribution to the
bound on the resistance. We start with the general case $d > 3$. We will discuss
the special cases $d=2$ and $d=3$ below. \\

In the computations below, by the ``contribution of a step'' we mean the
part of the sum in Theorem \ref{th-bollobas} that goes over the edges
considered in the current step. \\

{\bf Step 1} We start with a unit flow at $s$ that we send over all
$d_s$ adjacent edges. This leads to flow $1/d_s$ over $d_s$ edges. According to the formula
in Theorem \ref{th-bollobas} this contributes 

\ba r_1 = d_s \cdot \frac{1}{d_s^2} = \frac{1}{d_s}
\ea
to the overall resistance $R_{st}$. \\

{\bf Step 2:} After Step 1, the flow sits on all neighbors of $s$, and
these neighbors are not necessarily all contained in $C(s)$. To
proceed we want to re-concentrate all flow in $C(s)$. For each
neighbor $p$ of $s$, we thus carry the flow along a Hamming path of
cells from $p$ back to $C(s)$, see Figure \ref{fig-step2} for an
illustration.\\

To compute an upper bound for Step 2 we exploit that each neighbor $p$
of $s$ has to traverse at most $Q$ cells to reach $C(s)$ (recall the
definition of $Q$ from the proposition). Let us fix $p$. After Step 1,
we have flow of size $1/d_s$ in $p$. We now move this flow from $p$ to all points in the neighboring
cell $C(2)$ (cf. Figure \ref{fig-step2}). For this we can use at least $\nmin$ edges. Thus we send flow
of size $1/d_s$ over $\nmin$ edges, that is each edge receives flow
$1/ (d_s \nmin)$. Summing the flow from $C(p)$ to $C(2)$, for all points $p$, gives 
\ba
d_s \nmin \left( \frac{1}{d_s \nmin} \right)^2 =  \frac{1}{d_s
  \nmin}, 
\ea
Then we transport the flow from $C(2)$ along to $C(s)$. Between each
two cells on the way we can use $\nmin^2$
edges. Note, however, that we need to take into account that some of
these edges might be used several times (for different points $p$). In
the worst case, $C(2)$ is the same for all points $p$, in which 
case we send the whole unit flow over these edges. This amounts to
flow of size $1 / (\nmin^2)$ over
$(Q-1)\nmin^2$ edges, that is a contribution of 
\ba
\frac{Q-1}{\nmin^2}. 
\ea

Altogether we obtain 
\ba
r_2 \leq \frac{1}{d_s \nmin}
+ 
\frac{Q}{\nmin^2}
. \\
\ea

{\bf Step 3:}
At the beginning of this step, the complete unit flow resides in the cube
$C(s)$. We now want to distribute this flow to a ``hypercube'' of three
dimensions (no matter what $d$ is, as long as $d>3$) that is perpendicular to the line that connects $s$ and
$t$ (see Figure \ref{fig-overview}, where the case of $d=3$ and a
2-dimensional ``hypercube'' are shown). To distribute the flow to this cube we divide it into layers  (see Figure \ref{fig-layers}). Layer 0 consists
of the cell $C(s)$ itself, the first layer consists of all cells
adjacent to $C(s)$, and so on. 
Each side of Layer $i$ consists of 
\ba
l_i = (2i+1) 
\ea
cells. For the 3-dimensional cube, the 
number $z_i$ of grid cells in Layer $i$, $i \geq 1$,  is given as 
\ba
z_i = \underbrace{6 \cdot (2i-1)^2}_{\text{interior cells of the faces}} + 
\underbrace{12 \cdot (2i-1)}_{\text{cells along the edges (excluding corners)}}
+ 
\underbrace{8}_{\text{corner cells}} \;\; = \;\; 24 i^2 + 2
\ea

All in all we consider 
\ba
a = \left\lfloor h/ (2
  g\sqrt{d})\right\rfloor \leq \left\lfloor h/ (2
  (g-1)\sqrt{d})\right\rfloor
\ea
layers, so that the final layer has
diameter just a bit smaller than the bottleneck $h$. 
 We now distribute the flow stepwise through all
layers, starting with unit flow in Layer $0$. 
To send the flow from Layer $i-1$ to Layer $i$ we use two phases, see
Figure \ref{fig-step3} for details. In the ``expansion phase'' 3a(i) we transmit the flow from Layer $i-1$ to all adjacent cells in
Layer $i$. In the ``redistribution phase'' 3b(i) we then redistribute the flow in Layer $i$ to
achieve that it is uniformly distributed in Layer $i$. In all phases,
the aim is to use as many edges as possible. \\

{\em Expansion phase 3a(i)}. We
can lower bound the number of edges between Layer $i-1$ and Layer $i$ 
by $z_{i-1} \nmin^2$: each of the $z_{i-1}$ cells in Layer $i-1$ is
adjacent to at least one of the cells in Layer $i$, and each cell
contains at least $\nmin$ points. Consequently, we can upper bound the
contribution of the edges in the expansion phase  3a(i) to the
resistance by

\ba
r_{3a(i)} \leq z_{i-1} \nmin^2 \cdot \left( \frac{1}{z_{i-1} \nmin^2} \right)^2 =
\frac{1}{z_{i-1} \nmin^2}.\\
\ea

{\em Redistribution phase 3b(i)}. %
We make a crude upper bound for the redistribution phase. In this phase
we have to move some part of the flow from each cell to its neighboring
cells. For simplicity we bound this by assuming that for each cell, we had to move
{\em all} its flow to neighboring cells. By a similar argument as for Step 3a(i), the contribution of the
redistribution step can be bounded by

\ba
r_{3b(i)} \leq z_i \nmin^2 \cdot \left( \frac{1}{z_i \nmin^2} \right)^2 =
\frac{1}{z_i \nmin^2}.
\ea

{\em All of Step 3.} 
All in all we have $a$ layers. Thus the overall contribution of
Step 3 to the resistance can be bounded by 

\banum
r_3 & \;= \;\sum_{i=1}^{a} r_{3a(i)} + r_{3b(i)}  
\; \leq \;
\frac{2}{\nmin^2}\sum_{i=1}^{a}  \frac{1}{z_{i-1}} 
\; \leq \;\frac{2}{\nmin^2}
\left(  1 + \frac{1}{24}  \sum_{i=1}^{a-1} \frac{1}{i^2} 
\right) 
\; \leq \; \frac{3}{\nmin^2}
\eanum
To see the last inequality, note that the sum $\sum_{i=1}^{a-1} 1/i^2$ is a
partial sum of the over-harmonic series that converges to a constant
smaller than 2. \\ 

{\bf Step 4:} Now we transfer all flow in ``parallel cell paths'' from
$H(s)$ to $H(t)$. We have $(2a+1)^3$ parallel rows of cells going from
$H(s)$ to $H(t)$, each of them contains $d(s,t) / g$ cells. Thus all
in all we traverse $(2a+1)^3 \nmin^2 d(s,t) /g $ edges, and each edge
carries flow $1 / ((2a+1)^3 \nmin^2 )$. Thus step 4 contributes 

\ba
r_4 \leq 
(2a+1)^3 \nmin^2 \frac{ d(s,t)}{g}
\cdot \left(  \frac{1}{(2a+1)^3 \nmin^2}    \right)^2 
\;=\; 
\frac{d(s,t)}{g(2a+1)^3\nmin^2 } \\
\ea

{\bf Step 5} is completely analogous to steps 2 and 3, with the analogous 
contribution $r_5 = \frac{1}{d_t  \nmin} + r_3$. \\

{\bf Step 6} is completely analogous to step 1 with overall
contribution of $r_6 = 1 / d_t$. \\

{\bf Summing up the general case $d>3$.} All these contributions leads to the following
overall bound on the resistance in case $d>3$: 

\ba
R_{st} 
& \leq 
\frac{1}{d_s} + \frac{1}{d_t} + 
\left(
\frac{1}{d_s} + \frac{1}{d_t}
\right)
\frac{2}{\nmin}
+ 
\frac{1}{\nmin^2} 
\left(3   + \frac{d(s,t)}{g(2a+1)^3 } + 2Q\right)\\
\ea
with $a$ and $Q$ as defined in Proposition
\ref{prop-resistance-fixed-graph}. This is the result stated in the
proposition for case $d>3$. \\

Note that as spelled out above, the proof works whenever the dimension
of the space satisfies $d > 3$. In particular, note that even if $d$
is large, we only use a 3-dimensional ``hypercube'' in Step 3. It is
sufficient to give the rate we need, and carrying out the construction for
higher-dimensional hypercube (in particular Step 3b) is a pain that we
wanted to avoid.\\

{\bf The special case $d=3$.} In this case, everything works very similar to above, except that we we only
use a 2-dimensional ``hypercube'' (this is what we always
show in the figures). The only place in the proof where
this really makes a difference is in Step 3. The number $z_i$ of grid
cells in Layer $i$ is given as $z_i = 8i$. Consequently, instead of obtaining an
over-harmonic sum in $r_3$ we obtain a harmonic sum. Using the
well-known fact that $\sum_{i=1}^a 1/i \leq \log(a) + 1$ we obtain  

\ba
r_3 \leq 
\frac{2}{\nmin^2}
\left(  1 + \frac{1}{8}  \sum_{i=1}^{a-1} \frac{1}{i} 
\right) 
\; \leq \; 
\frac{2}{\nmin^2}
\left(  2 + \log(a)
\right) 
\ea

In Step 4 we just have to replace the terms $(2a+1)^3$ by
$(2a+1)^2$. 
This leads to the result in Proposition
\ref{prop-resistance-fixed-graph}.  \\

{\bf The special case $d=2$.} Here our ``hypercube'' only consists of
a ``pillar'' of $2a+1$ cells. The fundamental difference to higher
dimensions is that in Step 3, the flow does not have so much ``space'' to be
distributed. Essentially, we have to distribute all unit flow through
a ``pillar'', which results in  contributions 

\ba
& r_3 \leq \frac{2a+1}{\nmin^2}\\
& r_4 \leq \frac{d(s,t)}{g} \frac{1}{(2a+1) \nmin^2}\\
\ea

This concludes the proof of Proposition
\ref{prop-resistance-fixed-graph}. 
\hfill\smiley\\

Let us make a couple of technical remarks about this proof. 
For the ease of presentation we simplified the proof in a couple of
respects.  \\

Strictly speaking, we do not need to distribute the whole unit flow to
the outmost Layer $a$. The reason is that in each layer, a fraction of the
flow already ``branches off'' in direction of $t$. We simply ignore this
leaving flow when bounding the flow in Step 3, our construction leads
to an upper bound. It is not difficult to take the outbound flow
into account, but it does not change the order of magnitude of the final result. So for the
ease of presentation we drop this additional complication and stick to
our rough upper bound.\\

When we consider Steps 2 and 3
together, it turns out that we might have introduced some loops in the
flow. To construct a proper flow, we can simply remove these
loops. This would then just reduce the
contribution of Steps 2 and 3, so that our current estimate is an overestimation of the whole resistance. \\

The proof as it is spelled out above considers the case where $s$ and
$t$ are connected by a straight line. It can be generalized to
the case where they are connected by a piecewise linear path. This
does not change the result by more than constants, but adds some technicality at
the corners of the paths. \\

The construction of the flow only works if the bottleneck of $\Xcal$
is not smaller than the diameter of one grid cell, if $s$ and $t$ are at
least a couple of grid cells apart from each other, and if $s$ and $t$
are not too close to the boundary of $\Xcal$. We took care of
these conditions in Part 3 of the definition of a valid grid. \\

\subsection{Proof of the Theorems \ref{th-main-eps} and \ref{th-main-knn}}

First of all, note that by Rayleigh's principle  (cf. Corollary 7 in
  Section IX.2 of \citealp{Bollobas98}) the
effective resistance between vertices cannot decrease if we delete
edges from the graph. Given a sample from the underlying density
$p$, a random geometric graph based on this sample, and some valid
region $\Xcal$, we first delete all points that are not in
$\Xcal$. Then we consider the remaining geometric graph. The effective
resistances on this graph are upper bounds on the resistances of the
original graph. Then we conclude the proofs with the following
arguments: \\

{\bf Proof of Theorem \ref{th-main-eps}.}  The lower bound on the deviation follows
immediately from Proposition \ref{th-lower}. The upper bound is a
consequence of Proposition \ref{prop-resistance-fixed-graph} and well
known properties of random geometric graphs (summarized in the
appendix). In particular,
note that we can choose the grid width $g := {\eps} / ({2\sqrt{d-1}})$
to obtain a valid grid. The quantity $\nmin$ can be bounded as stated
in Proposition \ref{prop-counting} and is of order $n \eps^d$, the degrees behave as described
in Proposition \ref{prop-degrees-eps} and are also of order $n \eps^d$ (we use $\delta = 1/2$ in these
results for simplicity). The quantity $a$ in Proposition
\ref{prop-resistance-fixed-graph} is of the order $1/\eps$, and $Q$ can be bounded by $Q = \eps
/ g$ and by the choice of $g$ is indeed a constant. Plugging all these
results together leads to the final statement of the theorem. 
\hfill\smiley\\

{\bf Proof of Theorem \ref{th-main-knn}.} 
 This proof is analogous to
the $\eps$-graph. As grid width $g$ we choose 
$g=\rkmin / (2 \sqrt{d-1})$ where $\rkmin$ is the minimal
$k$-nearest neighbor distance  (note that this works for
both the symmetric and the mutual $\knn$-graph). 
Exploiting Propositions 
\ref{prop-counting} and
\ref{prop-degrees-knn} we can see that $\rkmin$ and $\rkmax$ are of
order $(k/n)^{1/d}$, the degrees and $\nmin$ are of
order $k$,  $a$ is of the order $(n / k)^{1/d}$ and $Q$ a constant. 
Now the statements of the
theorem follow from Proposition
\ref{prop-resistance-fixed-graph}. 
\hfill\smiley

\section{Proofs for the spectral approach} \label{sec-proofs-spectral}
\subsection{Proof of the key propositions \ref{prop-commute} and
  \ref{prop-approx}}\label{sec-proof-basic}

In this section we prove the general formulas to compute and
approximate the hitting times. \\

{\bf Proof of Proposition \ref{prop-commute}: }
For the hitting time formula, let $u_1,\ldots,u_n$ be an orthonormal
set of eigenvectors of $\lsym$ corresponding to the eigenvalues
$\lambda_1,\ldots,\lambda_n$.  Let $u_{ij}$ denote the $j$-th entry of
$u_i$.  According to \citet{Lovasz93} the hitting time is given by
\[H_{ij}=\vol(G) \sum_{k=2}^n\frac{1}{1-\lambda_k}\left(\frac{u_{kj}^2}{d_j}-\frac{u_{ki}u_{kj}}{\sqrt{d_id_j}}\right).\]
A straightforward calculation using the spectral representation  of $\lsym$ yields 
\begin{align*}
H_{ij}%
&=\vol(G)\ip{\tfrac{1}{\sqrt{d_j}}e_j}{\sum_{k=2}^n\frac{1}{1-\lambda_k}\ip{u_k}{\tfrac{1}{\sqrt{d_j}}e_j-\tfrac{1}{\sqrt{d_i}}e_i}\;u_k}\\
&=\vol(G)\ip{\tfrac{1}{\sqrt{d_j}}e_j}{\lsym\pinv\left(\tfrac{1}{\sqrt{d_j}}e_j-\tfrac{1}{\sqrt{d_i}}e_i\right)}. 
\end{align*}
The result for the commute time follows from the one for the hitting
times. 
\mbox{}\hfill \ulesqed\\[1cm]

In order to prove Proposition \ref{prop-approx} we first state a small lemma. 
For convenience,  we set $A = D^{-1/2}WD^{-1/2}$ and 
$u_i = e_i /
\sqrt{d_i}$.  Furthermore, we are going to denote the projection on
the eigenspace of the $j$-the eigenvalue $\lambda_j$ of 
$A$ 
by $P_j$.

\begin{lemma}[Pseudo-inverse $\lsym\pinv$] \label{prop-lsympinv}
The pseudo-inverse of the symmetric Laplacian satisfies
\ba
\lsym\pinv &= I - P_1 + M\ea
where $I$ denotes the identity matrix and $M$ is given as follows: 
\banum \label{eq-m1}
M \;\; = \;\;\sum_{k=1}^\infty (A-P_1)^k  \;\;= \;\;\sum_{r=2}^n \frac{\lambda_r}{1 - \lambda_r} P_r
\eanum
Furthermore, for all $u, v \in \R^n$ we have 
\banum \label{eq-m2}
| \inner{u , M v } |
\;\leq\;
\frac{1}{1 - \lambda_2} \cdot
\| (A-P_1)u \| \;\cdot
\; \|(A-P_1) v \| + 
|\inner{ u \;,\; (A-P_1) v }|
\eanum

\end{lemma}
\begin{proof}
The projection onto the
null space of $\lsym$ is given by $P_1 = \sqrt{d}\sqrt{d}^T / {\sum_{i=1}
  d_i} $ where $\sqrt{d}=(\sqrt{d_1},\ldots,\sqrt{d_n})^T$. 
As the graph is not bipartite, $\lambda_n
> -1$. Thus the 
pseudoinverse of $\Lsym$ can be computed as
\[  \Lsym\pinv = (I-A)\pinv = (I - A + P_1)^{-1} - P_1 = \sum_{k=0}^\infty (A-P_1)^k - P_1.\]

Thus 
\ba
M & : = \sum_{k=1}^\infty (A-P_1)^k 
\;\;=\;\; \sum_{k=0}^\infty (A - P_1)^k (A-P_1) \\
& = ( \sum_{k=0}^\infty 
 \sum_{r=2}^n \lambda_r^k P_r )
( \sum_{r=2}^n \lambda_r P_r )
\;\; = \;\; 
(\sum_{r=2}^n \frac{1}{1 - \lambda_r} P_r ) ( \sum_{r=2}^n \lambda_r P_r )\\
&  = \sum_{r=2}^n \frac{\lambda_r}{1 - \lambda_r} P_r
\ea
which proves Equation \eqref{eq-m1}. By a little detour, we can also see 
\ba
M %
\;\;=\;\; 
\sum_{k=0}^\infty (A-P_1)^k (A- P_1)^2+ (A - P_1) 
\;\;=\;\;
 (\sum_{r=2}^n \frac{1}{1 - \lambda_r} P_r ) (A- P_1)^2+ (A - P_1).\\
\ea
Exploiting that
$(A-P_1)$ commutes with all $P_r$ gives 
\ba 
\inner{u , M v } 
=  \big\langle 
(A-P_1) u
\;, 
\; 
(\sum_{r=2}^n \frac{1}{1 - \lambda_r} P_r ) (A- P_1) v 
\big\rangle
+ 
\inner{u ,  (A - P_1) v }. 
\ea

Applying the Cauchy-Schwarz
inequality and the fact $\|\sum_{r=2}^n \frac{1}{1 -
  \lambda_r} P_r\|_2 = 1 / (1 - \lambda_2)$  leads to the desired
statement.

\end{proof}

{\bf Proof of Proposition \ref{prop-approx}. }
This proposition now follows easily from the Lemma above. Observe that
\ba
& \inner{ u_i, A u_j }
= \frac{w_{ij}}{d_i d_j} \leq \frac{\wmax}{\dmin^2} \\
& \|A u_i \|^2
= \sum_{k=1}^n \frac{w_{ik}^2}{d_i^2 {d_k}}
\leq \frac{\wmax}{\dmin d_i^2} {\sum_{k} w_{ik}}
= \frac{\wmax}{\dmin} \frac{1}{d_i} 
\leq 
\frac{\wmax}{\dmin^2} \\
& \| A (u_i - u_j)\|^2  %
\leq 
\frac{\wmax}{\dmin} \left(\frac{1}{d_i} + \frac{1}{d_j}\right)
\leq 
\frac{2\wmax}{\dmin^2}\\
 & 
\ea
Exploiting
that $P_1(u_i- u_j) = 0$ we get for the hitting  time 
\banum
\left|
\frac{1}{\vol(G)} H_{ij} 
- \frac{1}{d_j}
\right|& = 
| \inner{ u_j , M (u_j - u_i) } | \nonumber\\
& \leq \frac{1}{1 - \lambda_2}
\|A u_j\| \cdot \| A (u_j - u_i)\| + 
| \langle u_j, A (u_j - u_i) \rangle | \nonumber\\
& \leq 
 \frac{1}{1 - \lambda_2}
\frac{\wmax}{\dmin}
\left( 
\frac{1}{\sqrt{d_j}} \sqrt{\frac{1}{d_i} + \frac{1}{d_j}}
\right)
+ \frac{w_{ij}}{d_i d_j} + \frac{w_{jj}}{d_j^2} \label{eq-genau}\\
& \leq \;\;
2 \frac{\wmax}{\dmin^2} \left( \frac{1}{1 - \lambda_2} + 1 \right) .
\eanum
For the commute time, we note that
\begin{align*} 
\left| \frac{1}{\vol(G)} C_{ij} 
- \Big(\frac{1}{d_i} + \frac{1}{d_j}\Big)
\right| & = \big| \inner{u_i - u_j, M (u_i-u_j)}\big| \\
        & \leq \frac{1}{1-\lambda_2}\norm{A(u_i-u_j)}^2 + \big| \inner{u_i - u_j, A (u_i-u_j)}\big| \\
        & \leq \frac{\wmax}{\dmin} \bigg(\frac{1}{1-\lambda_2} + 2\bigg)\bigg(\frac{1}{d_i} + \frac{1}{d_j}\bigg)
\end{align*}
\mbox{}\hfill\ulesqed\\

We would like to point out that the key to achieving this bound
is not to give in to the temptation to manipulate Eq. \eqref{eq-m1}
directly, but to bound Eq. \eqref{eq-m2}. The reason is that we can
compute terms of the form $\inner{u_i , A u_j}$ and related terms
explicitly, whereas we do not have any explicit formulas for the
eigenvalues and eigenvectors in \eqref{eq-m1}. 

\subsection{The spectral gap in random geometric
  graphs} \label{sec-technical}

As we have seen above, a key ingredient in the approximation result
for hitting times and commute distances is the spectral gap. In this
section we show how the spectral gap can be lower bounded for random
geometric graphs. We first consider the case of a fixed geometric
graph. From this general result we then derive
the results for the special cases of the $\eps$-graph and the
$\knn$-graphs. All graphs considered in this section are unweighted and
undirected. 
We follow the strategy in \citet{BoydEtal05} where the spectral gap is
bounded by means of the Poincar{\'e} inequality (see \citet{DiaStr91}
for a general introduction to this technique; see 
\citet{CooFri09} for a related approach in simpler
settings). The outline of this
technique is as follows: for each pair $(X,Y)$ of vertices in the graph we
need to select a path $\gamma_{XY}$ in the graph
that connects these two vertices. In our case, this selection is 
made in a random manner. Then we
need to consider all edges in the graph and investigate how many of the
paths $\gamma_{XY}$, on
average, traverse this edge. We need to control the maximum of this
``load'' over all edges. The higher this load is, the more pronounced
is the bottleneck in the graph, and the smaller the
spectral gap is. Formally, the spectral gap
is related to the maximum average load $b$ as follows.

\begin{proposition}[Spectral gap, 
  \citealp{DiaStr91}] \label{prop-boyd} \sloppy
Consider a finite, connected, undirected, unweighted graph
  that is not bipartite. For each pair of vertices $X \!\!\neq\!\!  Y$ let
  $P_{XY}$ be a probability distribution over all paths that connect
  $X$ and $Y$ and have
  uneven length. Let $(\gamma_{XY})_{X,Y}$ be a
  family of paths independently drawn from the respective
  $P_{XY}$. Define 
$b := \max_{\{e \text{ edge}\}} \mathbb{E} | \{ \gamma_{XY} \condon e \in
\gamma_{XY}\}|$. Denote by $|\gamma_{\max}|$
  the maximum path length (where the length of the path is the number
  of edges in the path).  
Then the spectral gap in the
  graph is bounded as follows: 
\banum
1 - \lambda_2 \geq
\frac{\vol(G)}{d_{\max}^2 |\gamma_{\max}| b}  
&& \text{ and } &&
1 - |\lambda_{n}|
\geq \frac{2}{ d_{\max} |\gamma_{\max}| b}.  \label{eq-gap}
\eanum
\end{proposition}

For deterministic sets
$\Gamma$, this proposition has been derived as Corollary 1 and 2 in
\citet{DiaStr91}. The adaptation for random selection of paths 
is straightforward, see \citet{BoydEtal05}.\\

The key to tight bounds based on Proposition \ref{prop-boyd} 
is a clever choice of the paths. We need to make sure that we
distribute the paths as ``uniformly'' as possible over the whole graph. 
This is relatively easy to achieve in the special situation where
$\Xcal$ is a torus with uniform distribution (as studied in
\citealp{BoydEtal05,CooFri09}) because of symmetry arguments and the
absence of boundary effects. However, in our setting with general
$\Xcal$ and $p$ we have to
invest quite some work. 

\subsubsection{Fixed geometric graph on the unit cube in $\R^d$}

\begin{figure}
\bc
\includegraphics[width=0.3\textwidth]{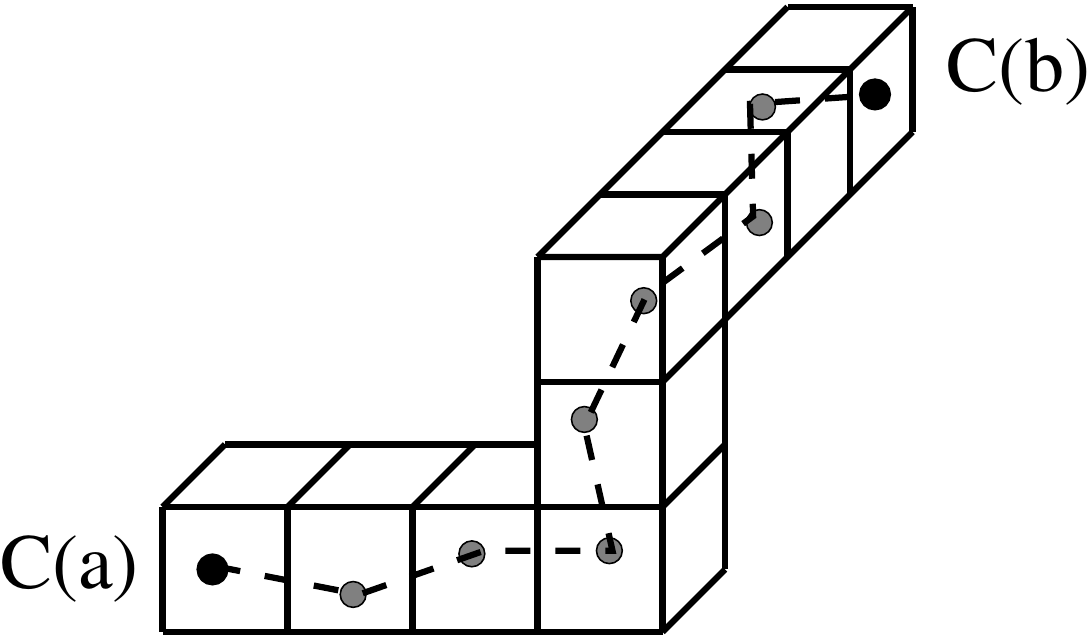}
\ec
\caption{Canonical path between $a$ and $b$. We first consider a
  ``Hamming path of cells'' between $a$ and $b$. In all intermediate
  cells, we randomly pick a point.  }
\label{fig-spectral}
\end{figure}

We first treat the special case of a fixed geometric graph with
vertices in the unit cube $[0,1]^d$ in $\R^d$. 
Consider a grid on the cube with grid width $g$. For now we assume
that the grid cells are so small that points in neighboring cells are
always connected in the geometric graph, and so large that each cell
contains a minimal number of data points. We will specify the exact 
value of $g$ later. In the following, cells of the grid are 
identified with their center points. \\

{\em Construction of the paths. } 
Assume we want to construct a path between two vertices $a$ and $b$
that correspond to the points $a = (a_1, \hd ,
a_d), \;b = (b_1, \hd, b_d) \in [0,1]^d$. Let $C(a)$ and $C(b)$ denote the
grid cells containing $a$ and $b$, denote the centers of these cells
by $c(a) = (c(a)_1, \hd, c(a)_d)$ and 
$c(b) = (c(b)_1, \hd,
c(b)_d)$. 
We first construct a deterministic ``cell path'' between
the cells $C(a)$ and $C(b)$ (see Figure \ref{fig-spectral}. This path simply follows a Hamming
path: starting at cell $C(a)$ we change
the first coordinate until we have reached $c(b)_1$. For example, if $c(a)_1 <
c(b)_1$  we
traverse the cells 
\ba
\big(c(a)_1, c(a)_2, \hd, c(a)_d \big) 
\leadsto \big (c(a)_1 + g, c(a)_2, \hd, c(a)_d \big)  
\leadsto
\hd \leadsto
\big (c(b)_1, c(a)_2, \hd, c(a)_d \big). 
\ea
Then we move along the second
coordinate from $c(a)_2$ until we have reached $c(b)_2$, that is we
traverse the cells $(c(b)_1, *, c(a)_3, \hd, c(a)_d)$. And so on. 
This gives a deterministic way of traversing adjacent cells from
$C(a)$ to $C(b)$. 
Now we transform this deterministic ``cell path'' to a random path on
the graph.  
In the special cases where $a$ and $b$ are in the same
cell or in neighboring cells, we directly connect $a$ and $b$ by an edge. In the general case, 
we select one data point uniformly at
random in each
of the interior cells on the cell path. Then we connect the selected
points to form a path.  Note that we can always force the paths to
have uneven lengths by adding one more point somewhere in between. 
\begin{proposition}[Path construction is valid] \label{prop-valid-path}
Assume that (1) Each cell of the grid contains at least one data point. 
(2) Data points in the same and in neighboring cells are always connected in the
graph. Then the graph is connected, and the paths constructed above
are paths in the graph. 
\end{proposition}
\begin{proof}
Obvious, by construction of the paths.
\end{proof}

In order to apply Proposition \ref{prop-boyd} we now need to compute the maximal average load of all paths. 

\begin{proposition}[Maximum average load for fixed graph on
  cube] \label{prop-load-cube} \sloppy
Consider a geometric graph on $[0,1]^d$ and the 
grid of width $g$ on $[0,1]^d$. Denote by $N_{\min}$ and $N_{\max}$
the minimal and maximal number of points per grid cell. Construct a
random set of paths as described above. 
\begin{enumerate}
\item 
Let $C$ be any fixed cell in the grid. Then there exist at most $d/
g^{d+1}$ pairs of cells $(A, B)$ such that cell paths starting in cell
$A$ and ending in cell $B$ pass through $C$. 

\item If the path construction is valid, then the maximal average load
  is upper bounded by 
\ba
b \leq
1 + \left( \frac{N_{\max}^2}{N_{\min}^2} +
2\frac{N_{\max}}{N_{\min}}
\right)
\frac{d}{g^{d+1}}.
\ea
\end{enumerate}
\end{proposition}
\begin{proof} {\em Part 1. }
  We identify cells with their centers. Consider two different grid
  cells $A$ and $B$ with centers $a$ and $b$. By construction, the
  Hamming path between $A$ and $B$ has the corners
\ba
a = &  (a_1, a_2, a_3, \hd, a_d) 
\leadsto 
(b_1, a_2, a_3, \hd, a_d) 
\leadsto 
(b_1, b_2, a_3, \hd, a_d) 
\\ & \leadsto \hd 
\leadsto 
(b_1, b_2, b_3, \hd,  b_{d-1}, a_d) 
\leadsto
(b_1, b_2, b_3, \hd,  b_{d-1}, b_d)  = b. 
\ea
All cells on the path have the form 
$
(b_1, b_2, \hd , b_{l-1}, *, a_{l+1}, \hd, a_d) 
$
where $*$ can take any value between $a_l$ and $b_l$. 
A path can only pass through the fixed cell with center $c$ if
there exists some $l \in \{1, \hd, d\}$ such that 
\ba
(c_1, \hd , c_d) 
= 
(b_1, b_2, \hd , b_{l-1}, *, a_{l+1}, \hd, a_d).  
\ea
That is, there exists some $l \in \{1, \hd, d\}$ such that 
\ba
(I) \;\; 
b_i = c_i 
\text{ for all } i=1, \hd, l-1 &&
\text{ and } && 
(II) \;\; a_i = c_i 
\text{ for all }i=l+1,\hd, d. 
\ea
For the given grid size $g$ there are $1/g$ different cell
centers per dimension. For fixed $l$ there thus exist 
$1/g^{d-l + 1}$ cell centers that satisfy $(I)$ and $1/g^{l}$ cell
centers that satisfy $(II)$. So all in all there are $1/g^{d+1}$ pairs
of cells $A$ and $B$ such that both $(I)$ and $(II)$ are satisfied for
a fixed value of $l$. Adding up the possibilities for all choices of $l
\in \{1, \hd, d\}$  leads to the factor $d$. \\

{\em Part 2. }  
Fix an edge $e$ in the graph and consider its two adjacent vertices
$v_1$ and $v_2$.
If $v_1$ and $v_2$ are in two different cells that are not neighbors to each other,
then by construction none of the paths traverses the edge.
If they are in the same cell, by construction at most one of
the paths can traverse this edge, namely the one
directly connecting the two points.
The interesting case is the one where $v_1$ and $v_2$ lie in two
neighboring grid cells $C$ and $\tilde C$.  \\
If both cells are
``interior'' cells of the path, then by construction each edge
connecting the two cells has equal probability of being selected. As
there are at least $N_{\min}$ points in each cell, there are at least
$N_{\min}^2$ different edges between these cells. Thus each of the
edges between the cells is selected with probability at most
$1/N_{\min}^2$. 
We know by Part 1 that there are at most $d/ {g^{d+1}}$ pairs of
start/end cells. As each cell contains at most $N_{\max}$ points, this
leads to $N_{\max}^2 d / {g^{d+1}}$ different paths passing through
$C$. This is also an upper bound on the number of paths passing
through both $C$ and $\tilde C$. Thus, each edge is selected by at
most 
$ d N_{\max}^2 /( g^{d+1} N_{\min}^2)$ paths. \\
If at least one of the cells is the start cell of the path, then
the corresponding vertex, say $v_1$, is the start point of the
path. If $v_2$ is an intermediate point, then it is selected with
probability at most $1/N_{\min}$ (the case where $v_2$ is an
end point has already been treated at the beginning). Similarly to the
last case, there are at most $N_{\max}d / g^{d+1}$
paths that start in $v_1$ and pass through $\tilde C$. This leads to an average load of
$ d N_{\max}/( g^{d+1} N_{\min})$ on edge $e$. The same holds with the
roles of $v_1$ and $v_2$ exchanged, leading to a factor 2. \\
The overall average load is now the sum of the average loads in the
different cases. 
\end{proof}

\subsubsection{Fixed geometric graph on a domain $\Xcal$ that is homeomorphic to a cube}

Now assume that $\Xcal \subset \R^d$ is a compact subset that is
homeomorphic  to the cube $[0,1]^d$ in the following sense: we assume that
there exists a homeomorphism $ h: \Xcal \to [0,1]^d$ and constants $0 < L_{\min} < L_{\max} <
\infty$ such that for all $x, y \in \Xcal$ we have
\banum \label{eq-homeo}
L_{\min} \| x-y \| \; \leq \; 
\|h(x) - h(y)\| 
\;\leq\;
L_{\max} \| x-y \|. 
\eanum

The general idea is now as follows. Assume we are given a geometric
graph on $X_1, \hd, X_n \in \Xcal$. In order to construct the paths we
first map the points in the cube using $h$. Then we construct the
paths on $h(X_1), \hd, h(X_n) \in [0,1]^d$ as in the last
section. Finally, we map the path back to $\Xcal$.

\begin{proposition}[Maximum average load for fixed graph on general domain] \label{prop-load-general}
\sloppy 
Let $G$ be a geometric graph based on $X_1, \hd, X_n \in
\Xcal$. Assume that there exists some $\tilde g > 0$ such that points
of distance smaller than $\tilde g$ are always connected in the
graph. Consider a mapping $h: \Xcal \to [0,1]^d$ as in
Equation~\eqref{eq-homeo} and a grid of width $g$ on $[0,1]^d$. 
Let $(C_i)_i$ be the cells of the $g$-grid on
  $[0,1]^d$, denote their centers by $c_i$. Let $B_i$ and $B_i'$ be 
  balls in $\Xcal$ with radius $r = g/(2L_{\max})$ and $R = \sqrt{d} \;g
  /L_{\min}$ centered at $h\inv(c_i)$. 
\be
\item These balls satisfy $B_i \subset h\inv(C_i)
  \subset B_i'$. 

\item Denote by $\tilde N_{\min}$ the minimal number of points
  in $B_i$ and $\tilde N_{\max}$ the maximal number of points in
  $B_i'$. Construct paths between the points $h(X_i) \in [0,1]^d$ as
  described in the previous subsection. 
If $\tilde N_{\min} \geq 1$ and $g \leq  L_{\min} \tilde g /
\sqrt{d+3}$, then these paths are valid. 

\item 
In this case, the maximal average load can be upper bounded by 
\banum \label{eq-load-fixed-general}
1 + \left( \frac{\tilde N_{\max}^2}{\tilde N_{\min}^2} +
2\frac{\tilde N_{\max}}{\tilde N_{\min}}
\right)
\frac{d}{(\tilde g L_{\min} / \sqrt{d+3})^{d+1}}. 
\eanum
\ee
\end{proposition}

\begin{proof} 
{\em Part 1.} Let $c_i$ be the center of cell $C_i$ and consider the ball $B_i$ centered at $h^{-1}(c_i)$ with radius $g/(2L_{\max})$. 
Clearly, $h^{-1}(c_i)$ is an interior point of $h^{-1}(C_i)$. 
Suppose that there exists $x\in B_i\cap \partial h^{-1}(C_i)$. 
Since $h$ maps the boundary of $h^{-1}(C_i)$ onto the boundary of $C_i$, we conclude that $h(x)\in \partial C_i$ and thus $\|h(x)-c_i\|\geq g/2$.  
By our assumption on the homeomorphism we can estimate 
\[\|x-h^{-1}(c_i)\|\geq \tfrac{1}{L_{\max}}\|h(x)-c_i\|\geq \tfrac{g}{2L_{\max}}.\]
Hence, $B_i\subseteq h^{-1}(C_i)$. 
To show the other statement let $x,y\in h^{-1}(C_i)$. 
Then \[\|x-y\| \leq \tfrac{1}{L_{\min}}\|h(x)-h(y)\|\leq \tfrac{1}{L_{\min}}\diam{C_i}=\tfrac{\sqrt{d}g}{L_{\min}}.\] 

{\em Part 2. } By the definition of $\tilde N_{\min}$ it is clear that
each cell of the grid contains at least one point. Consider two points
$X_i, X_j \in \Xcal$ such that $h(X_i)$ and $h(X_j)$ are in neighboring cells
of the $g$-grid. Then $\|h(X_i) - h(X_j)\| \leq g\sqrt{d+3}$. By
the properties of $h$, 
\ba
\|h\inv(X_i) - h\inv(X_j)\|
\leq
\frac{1}{L_{\min}}
\|X_i - X_j\| 
\leq \frac{1}{L_{\min}} \sqrt{d+3}\;g 
\;\;\;{\leq} \;\;\tilde g. 
\ea
Thus, by the definition of $\tilde g$ the points $X_i$ and $X_j$ are
connected in $G$. \\

{\em Part 3. } 
Follows directly from  
Proposition \ref{prop-load-cube}. 
\end{proof}

\subsubsection{Spectral gap for the $\eps$-graph}

Now we are going to apply 
Proposition \ref{prop-load-general} to 
$\eps$-graphs. We will use the general results on $\eps$-graphs
summarized in the appendix. 

\begin{proposition}[Maximal average load for $\eps$-graph] 
\label{prop-load-eps}
Assume that $\Xcal$ is homeomorphic to the cube with a mapping $h$ as
described in Equation $\eqref{eq-homeo}$. 
Then there exist constants $c_1, c_2, c_3 > 0$ such that with probability at least 
$ 1 - c_1 \exp(- c_2 n \eps^d) / \eps^d$, 
the maximum average load is upper bounded by $c_3 / \eps^{d+1}$. 
If $n \eps^d / \log n \to \infty$, then this probability tends to 1 as $n \to \infty$. 
\end{proposition}

\begin{proof} The proof is based on Proposition~\ref{prop-load-general}. By construction we know that points with
  distance at most $\tilde g = \eps$ are always connected in the $\eps$-graph. By
  Part 2 of Proposition~\ref{prop-load-general}, to 
ensure that points in neighboring grid cells are always connected in
the graph   we thus need to choose the grid width $g= \eps \cdot {L_{\min}}/{\sqrt{d+3}}
$. 
The radius $r$ defined in 
 Proposition~\ref{prop-load-general} is then given as 
\ba
r = \frac{g}{2L_{\max}} = 
\eps \cdot \frac{L_{\min}}{ 2\sqrt{d+3} L_{\max}}. 
\ea
The probability mass of the balls $B_i$ is thus bounded by  
\ba
b_{\min} \; \geq \; r^d \eta_d p_{\min} \alpha \; = \; 
\eps^d \; \cdot \; 
 \left(\frac{L_{\min}}{L_{\max}}\right)^d
 \frac{\eta_d}{2^d(d+3)^{d/2}} \; p_{\min} \alpha
\;\; = : \;\; \eps^d \cdot c_{\min}.  
\ea
We have 
\ba
K = 1/g^d = \sqrt{d+3} / L_{\min} ^d  \;\cdot \;(1/\eps^d) =: \kappa 
\;\cdot \;(1/\eps^d)
\ea
 grid
cells and thus the same number of balls $B_i$. We can now apply 
Proposition~\ref{prop-counting} (with $\delta: = 1/2$) to 
deduce the
bound for the quantity $\tilde N_{\min}$ used in 
Proposition~\ref{prop-load-general}: 
\ba
P\Big(
\tilde N_{\min} \;\leq \; n \eps^d c_{\min} / 2\Big)
\;\;\leq\;\; 
\frac{\kappa}{\eps^d}
\exp(- n \eps^d c_{\min} / 12 ).
\ea
Analogously, for $\tilde N_{\max}$ we have $R = \eps > \eps \sqrt{d} /
\sqrt{d+3}$ and $b_{\max} = R^d \eta_d p_{\max} = \eps^d
\eta_d p_{\max} := \eps^d \cdot c_{\max}$. With $\delta
= 0.5$ we then obtain 
\ba
P\Big(
\tilde N_{\max} \geq n \eps^d c_{\max} 3/2
\Big) 
\leq 
\frac{\kappa}{\eps^d}
\exp( - n \eps^d c_{\max} / 12 ).
\ea
Plugging these values into Proposition
\ref{prop-load-general} leads to the result. 
\end{proof}

We are now ready to {\bf prove Theorem \ref{th-gap-eps}} by applying
Proposition \ref{prop-boyd}. With probability at
least $1 - c_1 n \exp( - c_2 n \eps^d)$, 
both the minimal and maximal degrees in the graph
are of the order $\Theta( n\eps^d)$ (cf. Proposition
\ref{prop-degrees-eps}), and the volume of $G$ is of order
$\Theta(n^2 \eps^d)$.  To compute the maximal number
$|\gamma_{\max}|$ of edges in each of the
paths constructed above, observe that each path can traverse at most
$d \cdot 1/g = (d \sqrt{d+3} / L_{\min}) \cdot (1/\eps)$ cubes,
and a path contains just one edge per cube. Thus $| \gamma_{\max}|$ is
of the order $\Theta(1/\eps)$. 
In Proposition \ref{prop-load-eps} we have seen that with probability
at least $c_4 \exp(- c_5 n \eps^d) / \eps^d $ the maximum average load $b$ is of the order
$\Omega(1/\eps^{d+1})$. Plugging all these quantities in Proposition
\ref{prop-boyd}  leads to the result.  %
\mbox{}\hfill\ulesqed\\

\subsubsection{Spectral gap for the $\knn$-graph}

As in the case of the flow proofs, the techniques in the case of the $\knn$-graphs are
identical to the ones for the $\eps$-graph, we just have to replace
the deterministic radius $\eps$ by the minimal
$\knn$-radius. As before we exploit that  if two sample points have distance less
than $R_{k,\min}$ from each other, then they are always connected both
in
the symmetric and mutual $\knn$-graph.
\begin{proposition}[Maximal average load in the $\knn$-graph]
Under the general assumptions, with probability at least $1 - c_1
\cdot n \cdot \exp(-c_2 k)$  the
maximal average load in both the symmetric and mutual $\knn$-graph is bounded from above by $c_3 (n /
k)^{1 + 1/d}$. If $k / \log n \to \infty$, then this probability
converges to 1. 
\end{proposition} 
\begin{proof}
This proof is completely parallel to the one of Proposition
\ref{prop-load-eps}, the role of $\eps$ is now taken over by
$R_{k,\min}$. 
\end{proof}

Finally, the {\bf proof of Theorem \ref{th-gap-knn}} goes as follows. With probabilities
at least $1 - n \exp(- c_1 k)$ the following statements hold: the minimal and
maximal degree are  of order $\Theta(k)$, thus the number of edges in
the graph is of order $\Theta(n k)$. Analogously
to the proof for the $\eps$-graph, the maximal path length
$|\gamma_{\max}|$ is of the order $1 / R_{k,\min} = \Theta((k/n)^{1/d})$. The
maximal average load is of the order $O((n/k)^{d+1/d})$. Plugging all
these quantities in Proposition \ref{prop-boyd} leads to the result.
\mbox{}\hfill\ulesqed

\subsection{Proofs of Corollaries \ref{th-hitting-eps}  and \ref{th-hitting-knn} }
Now we collected all ingredients to finally present the following
proofs. \\

{\bf Proof of Corollary \ref{th-hitting-eps}} 

This  is a direct consequence of the results on the minimal degree
(Proposition \ref{prop-degrees-eps}) and the spectral gap (Theorem
\ref{th-gap-eps}). Plugging these
results into Proposition \ref{prop-approx} leads to the
first result. 
The last statement in the theorem follows by a standard density
estimation argument, as the degree of a
vertex in the $\eps$-graph is a consistent density estimator (see
Proposition \ref{prop-degrees-eps}). 
\mbox{}\hfill\ulesqed\\

{\bf Proof of Corollary \ref{th-hitting-knn} }

Follows similarly as Theorem \ref{th-hitting-eps} by applying Proposition 
\ref{prop-approx}. The results on the minimal degree and the spectral
gap can be found in Proposition \ref{prop-degrees-knn} and Theorem \ref{th-gap-knn}. 
The last statement follows from the convergence of the degrees, see
Proposition \ref{prop-degrees-knn}. 
\mbox{}\hfill\ulesqed\\

\subsection{Weighted graphs}

For weighted graphs, we use the following results 
from the literature. 

\begin{proposition}[Spectral gap in weighted graphs]
\label{prop-spectral-weighted}

\be
\item %
For any row-stochastic matrix $P$, 
\[ \lambda_2 \; \leq \; \frac{1}{2}\max_{i,j}\sum_{k=1}^n \big|
\frac{w_{ik}}{d_i} - \frac{w_{jk}}{d_j}\big| \; \leq \; 1- n
\min_{i,j} \frac{w_{ij}}{d_i} \; \leq \; 1 -
\frac{w_{\min}}{w_{\max}}.\]

\item %
Consider a weighted graph $G$ with edge weights $0 < \wmin \leq w_{ij}
\leq \wmax$ and denote its second eigenvalue by $\lambda_{2,
  weighted}$. Consider the corresponding unweighted graph where all
edge weights are replaced by $1$, and denote its second eigenvalue by
$\lambda_{2,unweighted}$. Then we have
\ba
( 1 - \lambda_{2,unweighted} )\cdot \frac{\wmin}{\wmax}
\;\;\leq\;\;
(1 - \lambda_{2,weighted})  
\;\;\leq\;\; (1 - \lambda_{2,unweighted}) \cdot \frac{\wmax}{\wmin}
\ea
\ee
\end{proposition}
\begin{proof} 
\be \item This bound was obtained by \citet{Zen72}, see also Section 2.5 of
\citet{Sen06} for a discussion. Note that the second inequality is
far from being tight. But in our application, both bounds lead to
similar results. 

\item 
This statement follows directly from the well-known representation
of the second eigenvalue $\mu_2$ of the normalized graph Laplacian $\Lsym$
(see Sec. 1.2 in \citealp{Chung97}),
\ba
\mu_{2} = 
\inf_{f \in \R^n}   \frac{\sum_{i,j=1}^n w_{ij} (f_i-f_j)^2 }{ \min_{c
    \in \R} \sum_{i=1}^n d_i (f_i-c)^2}. 
\ea
Note that the eigenvalue $\mu_2$ of the normalized Laplacian and the
eigenvalue $\lambda_2$ of the random walk matrix $P$ are in relation
$1 -\lambda_{2} = \mu_2$. 
\ee 
\end{proof}

We will now show to examples how this proposition can be used. The
first application of Proposition \ref{prop-spectral-weighted} is the
{\bf 
Proof of Theorem \ref{th:Hitting-fully-connected}}, which follows
directly from plugging in the first part of Proposition
\ref{prop-spectral-weighted} in Theorem \ref{prop-approx}.

The second application of Proposition \ref{prop-spectral-weighted} is
the following proof. \\

{\bf Proof of Theorem \ref{co:Gaussian}. }\\
We split 

\ba
\left|
n R_{ij} - \frac{1}{p(X_i)} - \frac{1}{p(X_j)} \right|
\leq 
\left|
nR_{ij} - \frac{n}{d_i} - \frac{n}{d_j} \right|
+
\left|
 \frac{n}{d_i} + \frac{n}{d_j}
- \frac{1}{p(X_i)} - \frac{1}{p(X_j)} \right|. 
\ea

Under the given assumption, the second term on the right hand side converges to 0 a.s. by
a standard kernel density estimation argument. The main work is the
first term on the right hand side. We treat upper and lower bounds of
$R_{ij} - 1/d_i - 1/d_j$ separately. \\

To get a lower bound, recall that by Proposition \ref{th-lower} we have 
\ba
R_{ij} \geq \frac{ Q_{ij}}{1 + w_{ij} Q_{ij} }
\ea
where $Q_{ij} = 1/(d_i - w_{ij}) + 1/(d_j - w_{ij})$ and $w_{ij}$ is
the weight of the edge between $i$ and $j$. It is
straightforward to see that under the given conditions,

\ba
n \left(
R_{ij} - \frac{1}{d_i} - \frac{1}{d_j} 
\right)
\geq 
n \left( 
\frac{ Q_{ij}}{1 + w_{ij} Q_{ij} } - \frac{1}{d_i} -
  \frac{1}{d_j}
\right) \to 0 \;\;a.s.
\ea

To treat the upper bound, we define the $\eps$-truncated Gauss graph
$G^{\eps}$ as the graph with edge
weights 
\[ w^\eps_{ij} := \left\{\begin{tabular}{ll} $w_{ij}$ & if $\norm{X_i-X_j}\leq \eps$,\\ $0$ & else.\end{tabular}\right.\]
Let $d^\eps_i = \sum_{j=1}^n w^\eps_{ij}$. Because of  $w^\eps_{ij}\leq w_{ij}$ and  Rayleigh's
principle, we have $R_{ij}\leq R^\eps_{ij}$, where $R^\eps$ denotes
the resistance of the $\eps$-truncated Gauss graph. Obviously, 
\begin{align*}
 n R_{ij} - \bigg(\frac{n}{d_i}+\frac{n}{d_j}\bigg) &
\leq \left|  n R^\eps_{ij} - \bigg(\frac{n}{d_i}+\frac{n}{d_j}\bigg)
\right| \\
& \leq 
 \underbrace{\left| nR^\eps_{ij} -
    \bigg(\frac{n}{d^\eps_i}+\frac{n}{d^\eps_j}\bigg)\right| }_{(*)}
 +  \underbrace{\left|\bigg(\frac{n}{d^\eps_i}+\frac{n}{d^\eps_j}\bigg)     -
                                                  \bigg(\frac{n}{d_i}+\frac{n}{d_j}\bigg)\right|}_{(**)}.
                                                   \\
\end{align*}

To bound term $(**)$ we show that the degrees
in the truncated graph converge to the ones in the non-truncated
graph. To see this, note that 
\begin{align*}
\Exp\Big( \frac{d^\eps_i}{n}\; \Big| \;X_i\Big) &= \frac{1}{(2\pi)^\frac{d}{2}}\frac{1}{h^d}\int_{B(X_i,\eps)} e^{\frac{-\norm{X_i-y}^2}{2h^2}} \, p(y)\, dy \\
                              &= \frac{1}{(2\pi)^\frac{d}{2}} \int_{B(0,\frac{\eps}{h})} e^{-\frac{\norm{z}^2}{2}}\, p(X_i + hz) \, dz\\
                              &= \Exp\Big(\frac{d_i}{n}\; \Big| \;X_i\Big) - \frac{1}{(2\pi)^\frac{d}{2}} \int_{\R^d \backslash B(0,\frac{\eps}{h})} e^{-\frac{\norm{z}^2}{2}}\, p(X_i + hz) \, dz.
\end{align*}
Exploiting that 
\begin{align*}
\frac{1}{(2\pi)^\frac{d}{2}} \int_{\R^d \backslash B(0,\frac{\eps}{h})} e^{-\frac{\norm{z}^2}{2}} 
& \leq \frac{1}{(2\pi)^\frac{d}{2}} e^{-\frac{\eps^2}{4h^2}} \int_{\R^d}  e^{-\frac{\norm{z}^2}{4}} \\
& \leq 2^\frac{d}{2} e^{-\frac{\eps^2}{4h^2}} = 2^\frac{d}{2} \frac{1}{\log(n\eps^{d+2})^\frac{1}{4}}
\end{align*}                       
we obtain the convergence of the expectations: under the assumptions on
$n$ and $h$ from the theorem, 
\[ \Big| \; \Exp\Big( \frac{d^\eps_i}{n}\; \Big|
\;X_i\Big) - \Exp\Big(\frac{d_i}{n}\; \Big| \;X_i\Big) \;\Big| \;\to\; 0.\]
Now, a probabilistic bound for term $(**)$
can be obtained by standard concentration arguments. \\

We now bound term $(*)$. In the following
we implicitly define $\eps$  via $h^2 = \eps^2 /
\log(n\eps^{d+2})$. Note that for the given choice of $\eps$, the
truncated Gaussian graph ``converges'' to the non-truncated graph, as
we truncate less and less weight. \\

Denote by $\lambda^{\eps, \text{weighted}}$ the eigenvalues of the
$\eps$-truncated Gauss graph, and by $\wmin^\eps$, $\wmax^\eps$ its
minimal and maximal edge weights. Also consider the graph $G''$ that
is the unweighted version of the $\eps$-truncated Gauss graph
$G^{\eps}$ . Note that $G''$ coincides with the standard $\eps$-graph. We denote its
eigenvalues by $\lambda^{\eps,
  \text{unweighted}}$.  
By applying Proposition
\ref{prop-approx} and Corollary \ref{th-hitting-eps} we get 
\banum
 \left| n R^\eps_{ij} - \bigg(\frac{n}{d^\eps_i}+\frac{n}{d^\eps_j}\bigg)\right| 
& \leq \frac{\wmax^\eps}{\dmin^\eps} \bigg(\frac{1}{1-\lambda^{\eps,\textrm{weighted}}_2} + 2\bigg)\bigg(\frac{n}{d^\eps_i} + 
\frac{n}{d^\eps_j}\bigg)\\ \nonumber
& \leq \frac{\wmax^\eps}{\dmin^\eps} \bigg(\frac{\wmax^\eps}{\wmin^\eps}\frac{1}{1-\lambda^{\eps,\textrm{unweighted}}_2} + 2\bigg)\bigg(\frac{n}{d^\eps_i} + \frac{n}{d^\eps_j}\bigg)\\
\label{eq-term-*}
\eanum
where the first inequality holds with probability 
 at
least $1 - c_1 n \exp(- c_2 n h^d) -c_3 \exp( - c_4 n h^d) /
 h^d$. By $(**)$ we already know that the last factor of Term
 \eqref{eq-term-*} converges to a constant: 
\ba
\bigg(\frac{n}{d^\eps_i} + \frac{n}{d^\eps_j}\bigg) \to 1/p(X_i) + 1/p(X_j)
\ea

For the other factors  of Term \eqref{eq-term-*} we use the following
quantities: 
\ba
& \wmin \geq \frac{1}{h^d} \exp(- \frac{\eps^2}{2h^2})
=  \frac{1}{h^d} \frac{1}{ (n\eps^{d + 2})^{1/2}}\\
& \wmax \leq \frac{1}{h^d}\\
& \dmin \geq n \eps^d \wmin\\
& 1 - \lambda_2 \geq \eps^2\\
\ea

Plugging these quantities  in \eqref{eq-term-*} we obtain the
convergence of $(*)$. 
\mbox{}\hfill\ulesqed\\

{\bf Proof of Corollary \ref{th-er}}

\begin{proof} 
It is well known that under the given assumptions, the following properties hold with high
probability: the graph is connected and the minimal and average degrees are of
the order $n p$, in particular $n p /d_i$ converges to 1 in
probability. The volume of the graph is of the order $n^2 p$. 
To use Theorem \ref{th-chungradcliffe}, observe that the 
matrix $\overline{A} = p J$ where $J$ is the $(n\times n)$-matrix of all ones. 
The expected degree of all vertices is $ np$. Hence, 
$\overline{D^{-1/2}}\overline{A}\overline{D^{-1/2}} = \frac{1}{np}
\cdot \overline{A}$. This matrix has rank 1, its non-zero eigenvalue
is 1 with the constant one vector as corresponding eigenvector. Hence 
the expected spectral gap in this model is 1. 
It is easy
to see that as soon as $p / \log(n) \to \infty$, the deviations in
Theorem \ref{th-chungradcliffe} converge to 0. 
Plugging all this into our
Proposition \ref{prop-approx} shows that with high
probability, 
\ba
np
\left|
\frac{1}{\vol(G)} H_{ij} - \frac{1}{d_i} \right|
 \;\leq\; 
np\cdot 
2 \left(\frac{1}{1 - \lambda_2} + 1\right) \frac{\wmax}{\dmin^2} 
= O\left(\frac{1}{n p}\right)\\
\ea
\end{proof}

{\bf Proof of Corollary \ref{th-er-planted}}

\begin{proof}
The expected degree of each vertex is $n (\pwithin+\pbetween)/2$, the expected volume of the
graph is $n^2 (\pwithin+\pbetween)/2$. 
 The matrix
$\overline{A}$ has the form 
$\left(\begin{smallmatrix} p J
  & q J\\ q J & pJ
\end{smallmatrix}\right)$
 where $J$ is the $(n/2 \times n/2)$-matrix of all
ones. The expected degree of all vertices is $ n(p+q)/2$. Hence, 
$\overline{D^{-1/2}}\overline{A}\overline{D^{-1/2}} = \frac{2}{n(p+q)}
\cdot \overline{A}$. This
matrix has rank 2, its largest eigenvalue is $1$ (with eigenvector
the constant 1 vector), the other eigenvalue is $(p-q) / (p+q)$ with
eigenvector 
$(1, ..., 1, -1, ..., -1)$. Hence, the spectral gap in this model is
$2q / (p+q)$. 

Under the assumption that $p = \omega(\log(n) / n)$, the deviations in
Theorem \ref{th-chungradcliffe} converge to 0. Plugging the expected
spectral gap in our bound in Proposition \ref{prop-approx} shows
that with high probability, 
\ba
& \frac{n (\pwithin+\pbetween)}{2}\cdot 
\left|
\frac{1}{\vol(G)} H_{ij} - \frac{1}{d_i} \right|
&\leq \frac{4}{n\pbetween} + \frac{4}{n(\pwithin+\pbetween)} = O\left(\frac{1}{n\pbetween}\right)
\ea
\end{proof}

{\bf Proof of Corollary \ref{th-expected-degrees}}

\begin{proof}
we use the result from Theorem 4 in
\citet{ChuRad11} which states that under the assumption that the
minimal expected degree $\overline{\dmin}$ satisfies
$\overline{\dmin}/\log(n) \to \infty$, then with probability at least
$1 - 1/n$ the spectral gap is bounded by a term of the order $O(
\log(2n) / \overline{\dmin} )$. 
Plugging this in Proposition \ref{prop-approx} shows
that with high probability, 

\ba
& \frac{ \left|\frac{1}{\vol(G)} C_{ij} - \frac{1}{d_i} -
  \frac{1}{d_j}\right|}{   \frac{1}{d_i} +
  \frac{1}{d_j} } 
& \leq
 \left( \frac{\overline{\dmin}}{log(2n)} + 2 \right) \frac{1}{\dmin}
= O\left(\frac{1}{\log(2n)}\right)
\ea
\end{proof}

\section{Discussion} \label{sec-discussion}

We have presented different strategies to prove that in many large
graphs the commute distance can be approximated by  $1/d_i +
1/d_j$. Both our approaches tell a similar story. Our result
holds as soon as there are ``enough disjoint paths'' between $i$ and
$j$, compared to the size of the graph, and the minimal degree is
``large enough'' compared to $n$. %

We would like to point out that our results on the degeneracy of the
hitting and commute times are not due to pathologies such as a 
``misconstruction'' of the graphs. For example, in the random
geometric graph setting  the graph Laplacian can be proved to
converge to the Laplace-Beltrami operator on the underlying space
under similar assumptions as the ones  above
\citep{HeiAudLux07}. But
even though the Laplacian itself converges to a meaningful limit, the
resistance distance, which is computed based on point evaluations of
the inverse of this Laplacian, does not converge to a useful limit. \\

The limit distance function $dist(i,j) = 1/d_i + 1/d_j$ is completely
meaningless as a distance function. It just considers the local
density (the degree) at the two vertices, but does not take into
account any global property such as the cluster structure of the
graph. As the speed of convergence is very fast (for example, of the
order $1/n$ in the case of Gaussian similarity graphs), the use of the raw
commute distance should be discouraged even on moderate sized
graphs. However, there might be ways how useful information can be
extracted from the commute distance, namely in the form of the
remainder terms $S_{ij} - 1/d_i - 1/d_j$.  Exploring this 
idea in depth is a project for
future research.\\

There are two important classes of graphs that are not covered in our
approach. In power law graphs as well as in grid-like graphs, the
minimal degree is constant, thus our results do not lead to
tight enough bounds. The resistance distances on grid-like graphs has been
studied in some particular cases.  For
example, \citet{Cserti00} and \citet{Wu04} prove explicit formulas for
the resistance on regular one-and two-dimensional grids, and
\citet{BenRos08} characterize the variance of the resistance on random
Bernoulli grids. To the best of our knowledge, general results about
the convergence of the resistance distance on grid-like graphs do not
exist. \\

\section{Appendix: General properties of random geometric graphs}

In this appendix we collect some basic results on random geometric
graphs. These results are well-known, but we did not find any reference where the
material is presented in the way we need it (often the results are
used implicitly or are tailored towards particular
applications). \\

In the following, assume that $\Xcal := \supp(p)$ is a valid region according to
Definition 1. Recall the definition of the boundary constant $\alpha$
in the valid region. 

A convenient tool for dealing with random geometric graphs is the
following well-known concentration inequality for binomial random variables that has first
appeared in \citet{AngVal77}. 

\begin{proposition}[Concentration inequalities]
\label{prop-concentration-appendix}
Let $N$ be a $Bin(n,p)$-distributed
random variable. Then, for all $\delta \in ]0,1]$, 
\ba P\Big(N\;\leq\; (1 - \delta) n p \Big) \;\;\leq\;\;
\exp(- \frac{1}{3} \delta^2 n p )\\
P\Big(N\;\geq\; (1 + \delta) n p \Big)
\;\;\leq\;\;
 \exp(- \frac{1}{3} \delta^2 n p ). 
\ea
\end{proposition}

We will see below that computing expected, minimum and maximum degrees
in random geometric graphs always boils down to counting the number of
data points in certain balls in the space. The following proposition
is a straightforward application of the concentration inequality above
and serves as ``template'' for all later proofs. 

\begin{proposition}[Counting sample points] \label{prop-counting}
Consider a sample $X_1, \hd, X_n$ drawn
i.i.d.\ according to density $p$ on $\Xcal$. Let $B_1, \hd, B_K$ be a fixed collection of subsets of $\Xcal$
  (the $B_i$ do not need to be disjoint). Denote by $b_{\min} :=
  \min_{i=1, \hd, K} \int_{B_i} p(x) dx$ the minimal probability mass
  of the sets $B_i$ (similarly by $b_{\max}$ the maximal probability
  mass), and by $N_{\min}$ and $N_{\max}$ the minimal (resp.\ maximal)
  number of sample points in the sets $B_i$. Then for all $\delta \in ]0,1]$
\ba
& P\Big(
N_{\max} \geq (1 + \delta) n b_{\max} \Big) \;
\leq \; 
K \cdot \exp( - \delta^2 n b_{\max} / 3
)\\
& P\Big(
N_{\min} \leq (1 - \delta) n b_{\min} \Big) \;
\leq \; 
K \cdot \exp( - \delta^2 n b_{\min} / 3
). 
\ea
\end{proposition}
\begin{proof} 
This is a straightforward application of Proposition
\ref{prop-concentration-appendix} using the union bound. 
\end{proof}

When working with $\eps$-graphs or $\knn$-graphs, we often need to
know the degrees of the vertices. As a rule of thumb, the
expected degree of a vertex in the $\eps$-graph is of the order
$\Theta(n \eps^d)$, the expected degree of a vertex in both the
symmetric and mutual $\knn$-graph is of the order $\Theta(k)$. The expected $\knn$-distance
is of the order $\Theta( (k/n)^{1/d})$. Provided the graph is
``sufficiently connected'' , all these rules of
thumb also apply to the minimal and maximal values of these quantities. The following
propositions make these rules of thumb explicit. 

\begin{proposition}[Degrees in the $\eps$-graph] \label{prop-degrees-eps}
Consider an $\eps$-graph on a valid region $\Xcal \subset \R^d$. 
\begin{enumerate} \item 
Then, for all $\delta \in ]0,1]$, the  minimal and maximal degrees in the $\eps$-graph satisfy 
\ba
& P\Big(
d_{\max} \geq (1 + \delta) n \eps^d p_{\max} \eta_d \Big) \;
\leq \; 
n \cdot \exp( - \delta^2 n \eps^d p_{\max} \eta_d/ 3
)\\
& P\Big(
d_{\min} \leq (1 - \delta) n \eps^d p_{\min} \eta_d \alpha \Big) \;
\leq \; 
n \cdot \exp( - \delta^2 n \eps^d p_{\min} \eta_d \alpha/ 3
). 
\ea
In particular, if $n \eps^d / \log n \to \infty$, 
then these probabilities converge to 0 as $n \to \infty$. 
\item If $n \to \infty, \eps \to 0$ and $n \eps^d / \log n \to
  \infty$, and the density $p$ is continuous, then for each interior point $X_i \in \Xcal$ the degree is a
  consistent density estimate: $d_i/ (n \eps^d \eta_d)  \longrightarrow p(X_i)$ a.s. 
\end{enumerate}
\end{proposition} 
\begin{proof} {\em Part 1} follows by applying 
Proposition~\ref{prop-counting} to the balls of radius
$\eps$ centered at the data points. Note that for the bound on
$d_{\min}$, we need to take into account boundary effects as only a
part of the $\eps$-ball around a boundary point is contained in
$\Xcal$. This is where the constant $\alpha$ comes in (recall the
definition of $\alpha$ from the definition of a valid region).  
{\em Part 2} is a standard density estimation argument: the expected degree
of $X_i$ is the expected number
of points in the $\eps$-ball around $X_i$. For $\eps$ small enough, the $\eps$-ball around $X_i$ is
completely contained in $\Xcal$ and the density is approximately constant on this ball because we assumed
the density to be continuous.  The expected
number of points is approximately $n \eps^d \eta_d p(X_i)$ where
$\eta_d$ denotes the volume of a $d$-dimensional unit ball. The result
now follows from Part 1. 
\end{proof}

Recall the definitions of the $k$-nearest neighbor radii: $R_k(x)$
denotes the distance of $x$ to its $k$-nearest neighbor among the $X_i$, and the
maximum and minimum values are denoted 
$R_{k,\max}:= \max_{i=1,...,n} R_k(X_i)$ and $R_{k,\min} :=
\max_{i=1,...,n} R_k(X_i)$. Also recall the definition of the boundary
constant $\alpha$ from the definition of a valid region. 

\begin{proposition}[Degrees in the   $\knn$-graph] 
\label{prop-degrees-knn}
\sloppy
Consider a valid region $\Xcal \subset \R^d$. 
\be
\item Define the constants $a = 1 / (2 p_{\max} \eta_d)^{1/d}$ and 
$\tilde a := 2 / ( \pmin \eta_d \alpha)^{1/d}$. Then

\ba
& P\Big(
R_{k, \min} \leq a \left(\frac{k}{n}\right)^{1/d}
\Big) 
\;\;\leq 
\;\;
n
 \exp(-k / 3)\\
& P\Big(
R_{k, \max} \geq \tilde a \left(\frac{k}{n}\right)^{1/d}
\Big) 
\;\;\leq 
\;\; n
 \exp(-k / 12). 
\ea
If $n \to \infty$ and $k / \log n \to \infty$, then these probabilities converge to 0.
\item 
Moreover, with probability  at least 
$ 1 - n \exp(- c_4 k)$ 
the minimal
  and maximal degree in both the symmetric and mutual $\knn$-graph are
  of the order $\Theta(k)$ (the constants differ). 
\item If the density is continuous, $n \to \infty$, $k / \log n \to
  \infty$ and additionally 
$k/n\rightarrow 0$,  then 
in both the symmetric and the mutual $\knn$-graph, the degree of any
fixed vertex $v_i$ in the interior of $\Xcal$ satisfies $k / d_i \to 1$ a.s. 
\ee
\end{proposition} 

\begin{proof}
{\em Part 1. } 
Define the constant $a = 1 / (2 p_{\max} \eta_d)^{1/d}$ and the radius 
$r := a \left({k}/{n}\right)^{1/d}$,  fix a sample
point $x$, and  denote by $\mu(x)$ the
probability mass of the ball around $x$ with radius $r$. Set 
$\mu_{\max} := r^d \eta_d p_{\max} \geq
\max_{x \in \Xcal} \mu(x) 
$. 
Note that  $\mu_{\max} < 1$. Observe that 
$ R_k(x) \leq r$ if and only if there are at
least $k$ data points in the ball of radius $r$ around $x$. 
Let $M \sim Bin(n, \mu)$ and $V \sim Bin(n, \mu_{\max})$. 
Note that by
the choices of $a$ and $r$ we have $E(V) = k/2$. All this leads to 
\ba
P
\Big(
R_k(x) \leq r
\Big) 
\;\;\leq\;\; 
 P
\Big(M \geq k
\Big)
\;\;\leq\;\; 
P\Big( V \geq k 
\Big)
\;\;= \;\; P\Big(V \geq 2 E(V) \Big). 
\ea 
Applying the concentration inequality of Proposition
\ref{prop-concentration-appendix} (with $\delta :=1)$) 
and using a union bound leads to the following result for the minimal $\knn$-radius: 
\ba
P\Big(
R_{k, \min} \leq a \left(\frac{k}{n}\right)^{1/d}
\Big) 
& \leq 
\;\;
P\Big(\exists i \; : \; 
R_{k}(X_i) \leq a \left(\frac{k}{n}\right)^{1/d}
\Big) \\
&\;\;\leq 
\;\;
n 
\max_{i =1,\hd,n} 
\Big(
R_k(X_i) \leq r
\Big) \\
&\leq
n
 \exp(-k / 3). 
\ea
By a similar approach we can prove the analogous statement for the 
maximal $\knn$-radius. Note that for the bound on $R_{k,\max}$ we additionally need to take
into account boundary effects:  at the boundary of $\Xcal$, only a part of the ball
around a point is contained in $\Xcal$, which affects the value of
$\mu_{\min}$. We thus define 
$\tilde a := 2 / ( \pmin \eta_d \alpha)^{1/d}$, $r := \tilde a
(k/n)^{1/d}$, 
$\mu_{\min} := r^d \eta_d p_{\min} \alpha$
where $\alpha \in ]0,1]$ is the constant defined in the valid
region. With $V = Bin(n,\mu_{\min})$ with $EV = 2k$ we continue
similarly to above and get (using $\delta = 1/2$)
\ba
P\Big(
R_{k, \max} \geq \tilde a \left(\frac{k}{n}\right)^{1/d}
\Big) 
\;\;\leq 
\;\; n
 \exp(-k / 12). 
\ea

{\em Part 2. } In the directed $\knn$-graph, the degree of each vertex
is exactly $k$. Thus, in the mutual $\knn$-graph, the maximum degree
over all vertices is upper bounded by $k$, in the symmetric $\knn$-graph the minimum
degree over all vertices is lower bounded by $k$.  

For the symmetric graph, observe that the maximal degree in the graph
is bounded by the maximal number of points in the balls of radius
$R_{k,\max}$ centered at the data points. We know that with high
probability, a ball of radius $R_{k,\max}$ contains of the order
$\Theta(n R_{k,\max}^d)$ points. Using Part 1 we know that with high
probability, $R_{k,\max}$ is of the order $(k/n)^{1/d}$. Thus the
maximal degree in the symmetric $\knn$-graph is of the order
$\Theta(k)$, with high probability.  

In the mutual graph, observe that the minimal degree in the graph is
bounded by the minimal number of points in the balls of radius
$R_{k,\min}$ centered at the data points. Then the statement follows
analogously to the last one.

{\em Part 3, proof sketch. } Consider a fixed point $x$ in the
interior of $\Xcal$.  
We know that both in the symmetric and mutual $\knn$-graph, two points
cannot be connected if their distance is larger than $R_{k,\max}$.  As we
know that $R_{k,\max}$ is of the order $(k/n)^{1/d}$, under the growth
conditions on $n$ and $k$ this radius becomes arbitrarily small. Thus,
because of the continuity of the density, if $n$ is large
enough we can assume that the density in the ball $B(x, R_{k,\max})$ of
radius $R_{k,\max}$ around $x$ is approximately constant. 
Thus, all points $y \in B(x, R_{k,\max})$ have approximately the same
expected $k$-nearest neighbor radius $R := (k / (n \cdot p(x)
\eta_d))^{1/d}$. Moreover, by concentration arguments it is easy to
see that the actual $\knn$-radii only
deviate by a factor $1 \pm \delta$ from their expected values.

Then, with high probability, all points inside of $B(x,R(1-\delta))$
are among the $k$ nearest neighbors of $x$, and all $k$ nearest
neighbors of $x$ are inside $B(x, R(1+\delta))$. On the other hand,
with high probability $x$ is among the $k$ nearest neighbors of all
points $y \in B(x,R(1-\delta))$, and not among the $k$ nearest
neighbors of any point outside of $B(x,R(1+\delta))$. 
Hence, in the mutual $\knn$-graph, with high probability $x$ is
connected exactly to all points $y \in B(x, R(1-\delta))$. In the
symmetric $\knn$-graph, $x$ might additionally be connected to the 
points in $B(x, R(1 + \delta)) \setminus B(x, R(1 - \delta))$. 
By construction, with high probability the number of sample points in these
balls is $(1+\delta)k$ and $(1 - \delta)k$. Driving $\delta$
to 0 leads to the result. 
\end{proof}

\end{document}

%% file: ules_latex_defaults.tex
\usepackage{amssymb}
\usepackage{amsmath}
\allowdisplaybreaks[3] 
\usepackage{bbold} 

\PassOptionsToPackage{pdftex}{graphicx} 

\usepackage[pdftex]{graphicx}

 \usepackage{ifthen} 
 \usepackage{layout} 
\usepackage{xspace}
\usepackage{afterpage}
\usepackage{url}
\usepackage{textcomp} 

\usepackage{wasysym} 

\usepackage[latin1]{inputenc} 

\def\ba#1\ea{\begin{align*}#1\end{align*}} 
\def\banum#1\eanum{\begin{align}#1\end{align}} 

\newcommand{\bi}{\begin{itemize}}
\newcommand{\ei}{\end{itemize}}
\newcommand{\be}{\begin{enumerate}}
\newcommand{\ee}{\end{enumerate}}
\newcommand{\bc}{\begin{center}}
\newcommand{\ec}{\end{center}}

\newtheorem{theorem}{Theorem} 
\newtheorem{lemma}[theorem]{Lemma} 
\newtheorem{proposition}[theorem]{Proposition} 

\newtheorem{corollary}[theorem]{Corollary}
\newtheorem{definition}[theorem]{Definition}

\newcommand{\ulesqed}{\hfill\smiley}

\ifx\proof\undefined
\newenvironment{proof}{\par\noindent{\em Proof. }}{\ulesqed\\[2mm]}
\else
\fi

% ...

\newcommand{\RR}{\mathbb{R}}

\let\R\undefined 
\newcommand{\R}{\RR}

\newcommand{\Xcal}{\mathcal{X}}

\DeclareMathOperator{\supp}{supp} 
\DeclareMathOperator{\diam}{diam}

\DeclareMathOperator{\vol}{vol} 

\newcommand{\norm}[1]{\|#1\|}

\newcommand{\pmin}{p_{\min}}
\newcommand{\pmax}{p_{\max}}
\newcommand{\dmin}{d_{\min}}

\newcommand{\wmin}{w_{\min}}
\newcommand{\wmax}{w_{\max}}
\newcommand{\rkmin}{R_{k,\min}}
\newcommand{\rkmax}{R_{k,\max}}

\newcommand{\Lsym}{L_{\text{sym}}}

\newcommand{\lsym}{\Lsym}

\newcommand{\knn}{\text{kNN}}

\newcommand{\eps}{\ensuremath{\varepsilon}}
\renewcommand{\epsilon}{\ensuremath{\varepsilon}}
\renewcommand{\phi}{\ensuremath{\varphi}}

\newcommand{\condon}{\; \big| \;}
\let\Pr\undefined 
\DeclareMathOperator{\Pr}{Pr}

\newcommand{\hd}{\hdots}

\newcommand{\inv}{^{-1}} 
\newcommand{\pinv}{^{\dagger}} 

\newcommand{\ER}{Erd\H{o}s-R\'{e}nyi\xspace} 
\newcommand{\er}{Erd\H{o}s-R\'{e}nyi\xspace}